\documentclass[3p,times]{elsarticle}

\usepackage{ecrc}

\volume{00}

\firstpage{1}

\journalname{Procedia Computer Science}

\runauth{}

\jid{procs}

\jnltitlelogo{Procedia Computer Science}

\CopyrightLine{2011}{Published by Elsevier Ltd.}

\usepackage{amssymb}
\usepackage[figuresright]{rotating}

\usepackage{afterpage}
\usepackage{makecell}
\usepackage{tabulary}
\usepackage{xcolor}
\usepackage{colortbl}
\usepackage{prettyref}
\usepackage{framed}
\usepackage{booktabs}       % professional-quality tables
\usepackage[utf8]{inputenc} % allow utf-8 input
\usepackage[T1]{fontenc}    % use 8-bit T1 fonts
\usepackage{hyperref}       % hyperlinks
\usepackage{url}            % simple URL typesetting
\usepackage{booktabs}       % professional-quality tables
\usepackage{amsfonts}       % blackboard math symbols
\usepackage{nicefrac}       % compact symbols for 1/2, etc.
\usepackage{microtype}      % microtypography
\usepackage{xcolor}         % colors
\usepackage{amsthm}
\usepackage{amsmath}
\usepackage{amssymb}
\usepackage{algorithm}
\usepackage{algpseudocode}
\usepackage{graphicx}
\usepackage{booktabs} 
\usepackage{bm}
\usepackage{hyperref}
\usepackage{cleveref}
\usepackage{enumitem}
\usepackage{mathtools}
\usepackage{array}
\usepackage{lipsum}
\usepackage{multirow}
\usepackage{makecell}
\usepackage{xcolor}
\usepackage{tikz}
\usepackage{etoolbox}
\newcommand{\pref}[1]{\prettyref{#1}}

\newcommand{\savehyperref}[2]{\texorpdfstring{\hyperref[#1]{#2}}{#2}}
\newrefformat{eq}{\savehyperref{#1}{Eq. \textup{(\ref*{#1})}}}
\newrefformat{eqn}{\savehyperref{#1}{Eq.~\textup{(\ref*{#1})}}}
\newrefformat{lem}{\savehyperref{#1}{Lemma~\ref*{#1}}}
\newrefformat{assump}{\savehyperref{#1}{Assumption~\ref*{#1}}}
\newrefformat{defi}{\savehyperref{#1}{Definition~\ref*{#1}}}\newrefformat{tab}{\savehyperref{#1}{Table~\ref*{#1}}}
\newrefformat{lemma}{\savehyperref{#1}{Lemma~\ref*{#1}}}
\newrefformat{line}{\savehyperref{#1}{Line~\ref*{#1}}}
\newrefformat{thm}{\savehyperref{#1}{Theorem~\ref*{#1}}}
\newrefformat{corr}{\savehyperref{#1}{Corollary~\ref*{#1}}}
\newrefformat{cor}{\savehyperref{#1}{Corollary~\ref*{#1}}}
\newrefformat{sec}{\savehyperref{#1}{Section~\ref*{#1}}}
\newrefformat{app}{\savehyperref{#1}{Appendix~\ref*{#1}}}
\newrefformat{ex}{\savehyperref{#1}{Example~\ref*{#1}}}
\newrefformat{fig}{\savehyperref{#1}{Figure~\ref*{#1}}}
\newrefformat{alg}{\savehyperref{#1}{Algorithm~\ref*{#1}}}
\newrefformat{rem}{\savehyperref{#1}{Remark~\ref*{#1}}}
\newrefformat{conj}{\savehyperref{#1}{Conjecture~\ref*{#1}}}
\newrefformat{prop}{\savehyperref{#1}{Proposition~\ref*{#1}}}
\newrefformat{proto}{\savehyperref{#1}{Protocol~\ref*{#1}}}
\newrefformat{prob}{\savehyperref{#1}{Problem~\ref*{#1}}}
\newrefformat{claim}{\savehyperref{#1}{Claim~\ref*{#1}}}
\newrefformat{que}{\savehyperref{#1}{Question~\ref*{#1}}}
\newrefformat{op}{\savehyperref{#1}{Open Problem~\ref*{#1}}}
\newrefformat{fn}{\savehyperref{#1}{Footnote~\ref*{#1}}}
\newrefformat{property}{\savehyperref{#1}{Property~\ref*{#1}}}

\usepackage{xspace}
\usepackage{natbib}
\usepackage{amsmath}
\usepackage{amsthm}
\usepackage{bbm}
\usepackage{mathrsfs}
\usepackage{bm}
\usepackage{makecell}
\usepackage{tabulary}

\usepackage[utf8]{inputenc}
\usepackage[english]{babel}

\newcommand{\calD}{\mathcal{D}}

\newcommand{\gray}[1]{{\color{gray}#1}}

\newtheorem*{theorem*}{Theorem}
\newtheorem{theorem}{Theorem}

\newtheorem{assumption}{Assumption}[section]
\newtheorem{lemma}[theorem]{Lemma}

\newtheorem*{lemma*}{Lemma}
\newtheorem{definition}{Definition}[section]

\theoremstyle{definition}

\newcommand*{\circled}[1]{\lower.7ex\hbox{\tikz\draw (0pt, 0pt)%
    circle (.5em) node {\makebox[1em][c]{\small #1}};}}
\robustify{\circled}

\newtheorem{thm}{Theorem}[section]

\usepackage{afterpage}
\usepackage{prettyref}
\usepackage{pgfplots}
\usepackage{caption}
\usepackage{amsmath,amssymb,amsfonts}
\usepgfplotslibrary{fillbetween}
\usepackage{framed}
\newrefformat{eq}{\savehyperref{#1}{Eq. \textup{(\ref*{#1})}}}
\newrefformat{eqn}{\savehyperref{#1}{Eq.~\textup{(\ref*{#1})}}}
\newrefformat{lem}{\savehyperref{#1}{Lemma~\ref*{#1}}}
\newrefformat{lemma}{\savehyperref{#1}{Lemma~\ref*{#1}}}
\newrefformat{defi}{\savehyperref{#1}{Definition~\ref*{#1}}}
\newrefformat{line}{\savehyperref{#1}{Line~\ref*{#1}}}
\newrefformat{thm}{\savehyperref{#1}{Theorem~\ref*{#1}}}
\newrefformat{corr}{\savehyperref{#1}{Corollary~\ref*{#1}}}
\newrefformat{cor}{\savehyperref{#1}{Corollary~\ref*{#1}}}
\newrefformat{sec}{\savehyperref{#1}{Section~\ref*{#1}}}
\newrefformat{table}{\savehyperref{#1}{Table~\ref*{#1}}}
\newrefformat{app}{\savehyperref{#1}{Appendix~\ref*{#1}}}
\newrefformat{assum}{\savehyperref{#1}{Assumption~\ref*{#1}}}
\newrefformat{ex}{\savehyperref{#1}{Example~\ref*{#1}}}
\newrefformat{fig}{\savehyperref{#1}{Figure~\ref*{#1}}}
\newrefformat{alg}{\savehyperref{#1}{Algorithm~\ref*{#1}}}
\newrefformat{rem}{\savehyperref{#1}{Remark~\ref*{#1}}}
\newrefformat{conj}{\savehyperref{#1}{Conjecture~\ref*{#1}}}
\newrefformat{prop}{\savehyperref{#1}{Proposition~\ref*{#1}}}
\newrefformat{proto}{\savehyperref{#1}{Protocol~\ref*{#1}}}
\newrefformat{prob}{\savehyperref{#1}{Problem~\ref*{#1}}}
\newrefformat{claim}{\savehyperref{#1}{Claim~\ref*{#1}}}
\newrefformat{que}{\savehyperref{#1}{Question~\ref*{#1}}}
\newrefformat{op}{\savehyperref{#1}{Open Problem~\ref*{#1}}}
\newrefformat{fn}{\savehyperref{#1}{Footnote~\ref*{#1}}}

\usepackage{xspace}
\usepackage{amsmath}
\usepackage{amsthm}
\usepackage{bbm}
\usepackage{mathrsfs}
\usepackage{amssymb}
\usepackage{bm}
\usepackage{makecell}
\usepackage{tabulary}
\usepackage{subcaption}

\theoremstyle{definition}

\begin{document}

\begin{frontmatter}

%% Title, authors and addresses

%% use the tnoteref command within \title for footnotes;
%% use the tnotetext command for the associated footnote;
%% use the fnref command within \author or \address for footnotes;
%% use the fntext command for the associated footnote;
%% use the corref command within \author for corresponding author footnotes;
%% use the cortext command for the associated footnote;
%% use the ead command for the email address,
%% and the form \ead[url] for the home page:
%%
%% \title{Title\tnoteref{label1}}
%% \tnotetext[label1]{}
%% \author{Name\corref{cor1}\fnref{label2}}
%% \ead{email address}
%% \ead[url]{home page}
%% \fntext[label2]{}
%% \cortext[cor1]{}
%% \address{Address\fnref{label3}}
%% \fntext[label3]{}

\renewcommand{\thefootnote}{\fnsymbol{footnote}}

\dochead{}
%% Use \dochead if there is an article header, e.g. \dochead{Short communication}
%% \dochead can also be used to include a conference title, if directed by the editors
%% e.g. \dochead{17th International Conference on Dynamical Processes in Excited States of Solids}

\title{Mitigating Data Reconstruction Attacks in Federated Learning: Conditions, Bounds, and Theoretical Guarantees}

%% use optional labels to link authors explicitly to addresses:
%% \author[label1,label2]{<author name>}
%% \address[label1]{<address>}
%% \address[label2]{<address>}

%\author{}

% \author[a]{Xiaojin Zhang}
% \author[b]{Yan Kang}
% \author[b]{Lixin Fan}
% \author[a]{Xueyang Wu}
% \author[a]{Kai Chen}
% \author[a,b]{Qiang Yang$^\dagger$}

% \affil[a]{Hong Kong University of Science and Technology, Hong Kong}
% \affil[b]{Peking University}
% \affil[b]{WeBank, China}

%\address{}

\author[label1]{Xiaojin Zhang
\footnote{School of Computer Science and Technology, Huazhong University of Science and Technology, China, email: xiaojinzhang@hust.edu.cn}}
\address[label1]{Huazhong University of Science and Technology, China}

\author[label1]{Mingcong Xu}

\author[label1]{Wei Chen
\footnote{School of Software Engineering, Huazhong University of Science and Technology, China}}

\author[label2]{Qiang Yang
\footnote{Corresponding Author, email: qyang@cse.ust.hk}}
\address[label2]{The Hong Kong Polytechnic University, Hong Kong, China}

\begin{abstract}
Privacy attacks, particularly data reconstruction attacks, pose significant challenges to federated learning despite its promise in privacy-preserving collaborative training. This paper presents a comprehensive theoretical analysis of privacy leakage in federated learning through dual perspectives: linear algebra and optimization theory. By analyzing the client-update Jacobian, we identify conditions
under which distinct data batches produce the same model update
and derive a sufficient batch-size criterion for non-unique
reconstruction. Through optimization theory, we establish a novel upper bound for privacy leakage that incorporates batch size, distortion extent, and other critical parameters. These theoretical insights elucidate the fundamental relationships between system parameters and privacy vulnerability, laying the groundwork for developing more robust privacy-preserving federated learning frameworks.
\end{abstract}


\end{frontmatter}

% \red{Need to discuss:}

% \red{1. We would be better off adding efficiency as part of the NFL because Prof. Yang cares about efficiency. He wants to use the NFL as a tool to prove that MPC/HE methods are not sufficient to protect privacy because their efficiency is low.}

% \red{2. need a formulation for the privacy-preserving LLM inference problem}

% \red{3. need applications of this NFL like our other NFL papers.}

\section{Introduction}
Federated learning \cite{mcmahan2017communication,yang2019federated} has gained significant attention in recent years as a distributed machine learning paradigm that enables multiple parties to collaboratively train a model without sharing their raw data. By keeping data locally and exchanging only model updates, federated learning reduces the need to share raw data and can help address privacy concerns. It has found applications in various domains, such as mobile computing \cite{hard2018federated}, healthcare \cite{antunes2022federated}, and finance \cite{long2020federated}.

Despite the promise of federated learning in protecting data privacy, recent studies have revealed that it is still susceptible to privacy attacks. Adversaries can exploit the shared model updates to infer sensitive information about the participants' private data. One notable class of attacks is data reconstruction attacks \cite{zhu2019deep,geiping2020inverting,yin2021see}, which aim to recover the original training data from the gradients or model updates. These attacks pose a severe threat to the privacy of federated learning participants and undermine the trust in the system.

To better understand and mitigate privacy risks in federated learning, it is crucial to conduct a rigorous theoretical analysis of privacy leakage. Previous works have investigated the privacy guarantees of federated learning from different perspectives, such as differential privacy \cite{dwork2006calibrating,abadi2016deep} and information theory \cite{wang2019beyond}. However, the theoretical understanding of privacy leakage in federated learning remains limited, especially in terms of the impact of the specific characteristics of federated learning in the local training process, such as the number of local data samples, the number of local epochs, and the batch size. 

In this paper, we aim to fill this gap by providing a theoretical analysis of privacy leakage in federated learning from two complementary perspectives: linear algebra and optimization theory. From the linear algebra perspective, we formulate the local training process as an optimization problem and examine the uniqueness of its solution. We prove that under the stated regularity assumptions, a Jacobian rank below
$Bp$ implies that distinct batches produce the same model
update, establishing local non-identifiability. We further derive a sufficient condition on the batch size to establish local non-identifiability. From the optimization theory perspective, we measure the privacy leakage using the discrepancy between the reconstructed data and the original data, and establish an upper bound on the privacy leakage in terms of the batch size, the distortion extent, and several other factors. Our analysis provides insights into the relationship between privacy leakage and various aspects of federated learning, such as the number of local data samples, the number of local epochs, and the batch size. Our main contributions are as follows:
\begin{itemize}
\item We formulate the local training process in federated learning as an optimization problem and analyze the uniqueness of its solution from the linear algebra perspective. We prove that under the stated regularity assumptions, a Jacobian
rank below $Bp$ implies that distinct batches produce the same
model update, establishing local non-identifiability.
\item We derive a sufficient condition on the batch size to establish local non-identifiability (\pref{thm: not_uniquely_determined} and \pref{thm: quantitative_dbp}), providing a guideline for setting the batch size in federated learning to enhance privacy protection.
\item We measure the privacy leakage using the discrepancy between the reconstructed data and the original data, and introduce the concept of distortion extent to quantify the gradient mismatch introduced by the privacy protection mechanism. We establish a connection between the privacy leakage, the batch size, the distortion extent, and other factors such as the number of local data samples and the number of local epochs. This reveals the relationship between privacy leakage and various aspects of federated learning (\pref{thm:privacy_leakage_upper_bound}).
\end{itemize}

The rest of the paper is organized as follows. Section \ref{sec:related_work} reviews the related work on federated learning, privacy attacks, and defense mechanisms. Section \ref{sec:preliminaries} introduces the preliminaries and notations used in the paper. Section \ref{sec:linear_equations} presents the theoretical analysis from the perspective of linear algebra, deriving conditions for establishing local non-identifiability. Section \ref{sec:average_case_optimization} provides the theoretical analysis from the perspective of  optimization, establishing upper bounds on privacy leakage. Section~\ref{experiment} presents the experimental evaluation of the effects of batch size on reconstruction leakage and model performance. Finally, Section \ref{sec:conclusion} concludes the paper and outlines future research directions.

This work sheds light on the theoretical aspects of privacy leakage in federated learning and provides insights into the design of privacy-preserving federated learning algorithms. By understanding the impact of different factors on privacy leakage, we can develop more effective defense strategies and build more secure and trustworthy federated learning systems.
\section{Related Work}
\label{sec:related_work}
In this section, we review the related work on federated learning, privacy attacks, and defense mechanisms.
\subsection{Federated Learning}
Federated learning has attracted significant attention in recent years as a promising approach for collaboratively training machine learning models while preserving data privacy. The concept of federated learning was first proposed by \citet{mcmahan2017communication}, who introduced the FedAvg algorithm for aggregating local models. Since then, various extensions and improvements have been proposed, such as FedProx \cite{li2020federated}, and FedNova \cite{wang2020tackling}. Furthermore, the concept of trustworthy federated learning has emerged, focusing on ensuring high utility, fairness, robustness, and efficiency, while preserving privacy. Researchers have explored the trade-offs between these factors and proposed various techniques to achieve a balance \cite{girgis2021shuffled, zhang2022no, mitchell2022optimizing, zhang2023towards, zhang2023trading, zhang2024meta, zhang2023game, he2024reinforcement, zhang2024deciphering}. Federated learning represents a promising direction for collaborative model training that can benefit a wide range of applications, from healthcare to finance, while addressing the growing concerns around data privacy and security.

\subsection{Privacy Attacks in Federated Learning}
Despite the privacy-preserving nature of federated learning, it has been shown that the shared model updates can still leak sensitive information about the participants' private data. \citet{zhu2019deep} proposed the Deep Leakage from Gradients (DLG) attack, which reconstructs the training data from the gradients by solving an optimization problem. \citet{geiping2020inverting} developed a cosine similarity-based attack called Inverting Gradients (IG) that achieves better reconstruction quality. \citet{yin2021see} introduced the GradInversion attack, which recursively reconstructs the private data from the model updates in multiple rounds.
Other notable privacy attacks in federated learning include membership inference attacks \cite{nasr2019comprehensive}, property inference attacks \cite{melis2019exploiting}, and model inversion attacks \cite{fredrikson2015model}.
\subsection{Defense Mechanisms}
To mitigate privacy risks in federated learning, various defense mechanisms have been proposed. Differential privacy \cite{dwork2006calibrating} is a well-established framework for protecting individual privacy by adding noise to the shared information. \citet{abadi2016deep} and \citet{mcmahan2017communication} applied differential privacy to federated learning by perturbing the gradients before aggregation. Secure aggregation \cite{bonawitz2016practical} is another approach that uses cryptographic techniques to ensure that the server can only see the aggregated model update without learning individual updates. Gradient compression \cite{haddadpour2021federated, albasyoni2020optimal} reduces the communication overhead and protects privacy by compressing the gradients before transmission.
Other defense mechanisms include participant-level differential privacy \cite{geyer2017differentially}, model perturbation \cite{wu2020value}, and gradient sparsification \cite{han2020adaptive}.

\section{Preliminaries}
\label{sec:preliminaries}
In this section, we introduce the basic concepts and notations used throughout the paper. Table \ref{tab:notation} summarizes the key notations.
\begin{table}[h]
\centering
\caption{Notation Table}
\label{tab:notation}
\begin{tabular}{c|l}
\toprule
Notation & Description \\
\midrule
$K$ & Number of clients \\
$\mathcal{D}^{k}$ & Local dataset of client $k$ \\
$n^k$ & Number of data samples in $\mathcal{D}^{k}$ \\
$\mathbf{x}_i^{k}, y_i^{k}$ & The $i$-th data sample and label of client $k$ \\
$\theta_t$ & Global model at round $t$ \\
$\theta_t^k$ & Local model of client $k$ at round $t$ \\
$\Delta \theta_{t}^k$ & Local model update of client $k$ at round $t$ \\
$\ell^{k}$ & Local objective function of client $k$ \\
$T$ & Number of communication rounds \\
$E$ & Number of local epochs \\
$B$ & Local batch size \\
$\eta$ & Learning rate \\
$\epsilon_p^{k}$ & Privacy leakage of client $k$ \\
$\Delta$ & Distortion extent \\
\bottomrule
\end{tabular}
\end{table}

\subsection{Federated Learning}
Federated learning \cite{mcmahan2017communication,yang2019federated} is a distributed machine learning paradigm that enables multiple parties to collaboratively train a model without sharing their raw data. In this paper, we focus on horizontal federated learning (HFL), where participants have different data samples but share the same feature space. HFL is widely used in real-world applications.

Consider a horizontal federated learning system with $K$ clients, where each
client $k$ has a private dataset
$D_k=\{(x_i^k,y_i^k)\}_{i=1}^{n_k}$, with $n_k=|D_k|$.
In HFL, $K$ participants collaboratively optimize a global model with
parameter $\theta$ by minimizing the weighted global objective:
\begin{equation}
\label{equ:fl}
    \min_{\theta}\ \ell(\theta)
    \triangleq
    \sum_{k=1}^{K}\frac{n_k}{N}\ell^k(\theta),
    \qquad
    N=\sum_{k=1}^{K} n_k,
\end{equation}
where
\[
\ell^k(\theta)
=
\frac{1}{n_k}
\sum_{i=1}^{n_k}
\ell(\theta,x_i^k,y_i^k)
\]
is the local objective function of client $k$, and $\ell(\cdot)$ is the loss function.

The widely adopted algorithm for solving the above optimization problem is
Federated Averaging (FedAvg)~\cite{mcmahan2017communication}.
In each communication round $t$, the server sends the current global model
$\theta_t$ to all clients. Each client $k$ then performs $E$ local epochs of
training on its local dataset $D_k$ to obtain the local model
$\theta_t^k$. The local model update
$\Delta\theta_t^k=\theta_t^k-\theta_t$
is sent back to the server, which aggregates the updates as
\begin{equation}
    \theta_{t+1}
    =
    \theta_t
    +
    \sum_{k=1}^{K}
    \frac{n_k}{N}\Delta\theta_t^k.
\end{equation}

\subsection{Privacy Attacks in Federated Learning}
Despite the promise of federated learning in protecting data privacy, recent studies have shown that it is vulnerable to various privacy attacks. In this paper, we focus on data reconstruction attacks \cite{zhu2019deep,geiping2020inverting,yin2021see}, which aim to recover the original training data from the shared model updates.
We consider a semi-honest adversary who follows the protocol faithfully but tries to infer sensitive information from the received messages. The adversary's goal is to reconstruct the private data of a target client $k$ by solving the following optimization problem:
\begin{equation}
\min_{\tilde{\mathbf{x}}, \tilde{y}} \textsc{Dist}\left(\textsc{Grad}(\theta_t, \tilde{\mathbf{x}}, \tilde{y}), \Delta \theta_{t}^k\right),
\end{equation}
where $\tilde{\mathbf{x}}$ and $\tilde{y}$ are the reconstructed data, $\textsc{Dist}(\cdot)$ is a distance metric, and $\textsc{Grad}(\cdot)$ simulates the local training process to approximate the real model update $\Delta \theta_{t}^k$.

\section{Theoretical Analysis from the Perspective of Linear Algebra}\label{sec:linear_equations}

In this section, we analyze the privacy of federated learning from the perspective of linear algebra. By formulating the local training process as an optimization problem and examining the uniqueness of its solution, we establish a theoretical basis for understanding how the model update relates to the identifiability of private data. Central to our analysis is the Jacobian matrix, which captures the sensitivity of the model update to changes in the input data. We prove that under the stated regularity assumptions, a Jacobian
rank below $Bp$ implies that distinct batches produce the same
model update, establishing local non-identifiability. Building upon this result, we derive a sufficient condition on
the batch size to establish local non-identifiability. This linear algebra based approach provides a rigorous and quantitative framework to reason about privacy in federated learning.

To solve the optimization problem in Eq.~(\ref{equ:fl}), federated learning \cite{mcmahan2017communication} involves $T$ iterations (i.e., communication rounds) of a training procedure between the server and $K$ clients, as illustrated in Algorithm~\ref{server_alg}. In each iteration $t$, the server sends the current version of the global model $\theta_t^g$ to all clients. Each client $k$ then computes multiple updates on its local model based on its private data and sends an updated version of the local model $\theta_{t}^k$ back to the server (see Algorithm \ref{client_update_alg}), which in turn aggregates the local models of all clients to form the next version of the global model. For simplicity of algorithm description, we assume the same learning rate $\eta$, batch size $B$, and number of training epochs $E$ are used for all clients, even though in practice these hyperparameters may be different across clients to account for their varying data distributions and computational capabilities.

\begin{algorithm}[!ht] 
\caption{FedAvg}
\label{server_alg}
\begin{algorithmic}[1]
\Statex \textbf{Input:} The total number of clients $K$; communication round $T$; learning rate $\eta$; batch size $B$; training epoch $E$.

\Statex \textbf{Initialize} global model parameter $\theta_0^g$

\For{round $t \in \{1,...,T\}$}
\For{each client $k$}

$\theta^k_{t+1} \leftarrow$ ClientUpdate $(\calD^k, \theta_t^g, \eta, B, E)$ \gray{$\triangleright$ \textit{call Algorithm \ref{client_update_alg}}}

\EndFor

\State $\theta_{t+1}^{g} \leftarrow \sum_{k=1}^{K} \frac{|\mathcal{D}^{k}|}{N}\theta_{t+1}^{k}$

\EndFor

\end{algorithmic}
\end{algorithm}

\begin{algorithm}[!ht] 
\caption{ClientUpdate}
\label{client_update_alg}
\begin{algorithmic}[1]

\Statex \textbf{Input:} 
client $k$'s dataset $\calD^k = \{\mathbf{x}^k, y^k\}$;
global model $\theta^g$; learning rate $\eta$; batch size $B$; training epoch $E$.
\State $m \leftarrow \left\lceil \frac{|\mathcal{D}^k|}{B} \right\rceil$
\State $\theta_{0,0}^k \leftarrow \theta^g$
\For{epoch $e \in \{1,...,E\}$}
    \State $\theta_{e,0}^k \leftarrow \theta_{e-1, m}^k$
    \State $\{\mathbf{x}_{e,b}^k, y_{e,b}^k\}_{b = 1}^{m} \leftarrow$ \text{PartitionData}($\mathbf{x}^k, y^k, B$)
    \For {batch $b \in \{1,...,m\}$}

    \State $\theta_{e,b}^k \leftarrow \theta_{e,b-1}^k - \eta \nabla _{\theta} \ell(\theta_{e,b-1}^k, \mathbf{x}_{e,b}^k, y_{e,b}^k) $
    \EndFor
\EndFor
\State Send $\theta_{E,m}^k$ to server
\end{algorithmic}
\end{algorithm}

In the FedAvg algorithm (Algorithm \ref{server_alg}), the server initializes the global model $\theta_0$. Then, for each communication round $t$, the server sends the current global model $\theta_t^g$ to each client $k$. Each client updates its local model using the ClientUpdate algorithm (Algorithm \ref{client_update_alg}) and sends the updated model $\theta^k_{t+1}$ back to the server. The server then aggregates the updated models from all clients by taking a weighted average based on the size of each client's local dataset, producing the new global model $\theta_{t}^g$.
In the ClientUpdate algorithm (Algorithm \ref{client_update_alg}), each client $k$ receives the global model $\theta^g$ from the server. Let $\theta_{e,b}^k$ represent the model parameters of client $k$ at epoch $e$ and batch $b$. The client initializes its local model $\theta_{0,0}^k$ with $\theta^g$. The client then performs $E$ local epochs of training. In each epoch $e$, the client partitions its local dataset $\{\mathbf{x}^k, y^k\}$ into batches of size $B$. For each batch $b$, the client updates its local model $\theta_{e,b}^k$ using the gradient of the loss function $\ell$ with respect to the model parameters, using a learning rate of $\eta$. After completing all local epochs, the client sends its updated model $\theta_{E,m}^k$ back to the server.

In round $t$ of training, the server distributes the global model parameters $\theta_t$ to all clients. Client $k$ randomly samples a batch of data of size $B$ from $\mathcal{D}^k$, denoted as $\{ \mathbf{x}_b^k, y_b^k \}_{b=1}^B$, and performs $E$ local epochs of training to update the model parameters to $\theta_{t}^k$. The model update $\Delta \theta_{t}^k = \theta_{t}^k - \theta_t$ is then uploaded to the server. The server aggregates all updates to obtain the new global model:
\begin{equation}
\theta_{t+1}
=
\theta_t
+
\sum_{k=1}^{K}
\frac{|\mathcal{D}^{k}|}{N}
\Delta\theta_t^{k}.
\end{equation}

For clarity, the following explicit gradient-sum formulation illustrates the
special case in which the same target mini-batch is used for $E$ local
updates. For the general \textsc{ClientUpdate} procedure in Algorithm~2,
the subsequent theoretical results are stated in terms of the full local-update
mapping $F(X)$, where $X$ denotes the target mini-batch and all other local
data, training configurations, and auxiliary randomness are held fixed.
Accordingly, the identifiability analysis below does not require the full
client update to admit the explicit fixed-batch gradient-sum representation.
% \begin{equation*}
% \begin{aligned}
% &\min_{\{\mathbf{x}_b^k,y_b^k\}_{b=1}^B} \quad  \frac{1}{B}\sum_{b=1}^B\ell(\theta_{t,e}^k, \mathbf{x}_b^k, y_b^k) \\
% &\textrm{s.t.} \quad  \sum_{e=0}^{E-1} \nabla_\theta\left(\frac{1}{B}\sum_{b=1}^B\ell(\theta_{t,e}^k, \mathbf{x}_b^k, y_b^k)\right) = \Delta \theta_{t}^k \\
% & \theta_{t,e+1}^k = \theta_{t,e}^k - \frac{\eta}{B} \sum_{b=1}^B\nabla_\theta\ell(\theta_{t,e}^k, \mathbf{x}_b^k, y_b^k), \quad e=0,\ldots,E-1 \\
% & \theta_{t,0}^k = \theta_t
% \end{aligned}
% \end{equation*}

\begin{equation*}
\begin{aligned}
\min_{\{x_b^k,y_b^k\}_{b=1}^{B}}
&\quad
\frac{1}{B}\sum_{b=1}^{B}
\ell(\theta_{t,e}^k,x_b^k,y_b^k)
\\
\text{s.t.}\quad
&
-\eta\sum_{e=0}^{E-1}
\nabla_{\theta}
\left(
\frac{1}{B}\sum_{b=1}^{B}
\ell(\theta_{t,e}^k,x_b^k,y_b^k)
\right)
=
\Delta\theta_t^k,
\\
&
\theta_{t,e+1}^k
=
\theta_{t,e}^k
-
\frac{\eta}{B}
\sum_{b=1}^{B}
\nabla_\theta
\ell(\theta_{t,e}^k,x_b^k,y_b^k),
\\
&
\theta_{t,0}^k=\theta_t.
\end{aligned}
\end{equation*}

where $\theta_{t,e}^k$ represents the local model parameters after the $e$-th epoch in round $t$, and $\eta$ is the learning rate.

In the above problem, the private data $(\mathbf{x}^k, y^k)$ is the optimization variable, and the constraint requires that the model update resulting from the data after $E$ local epochs should be equal to the observed $\Delta \theta_{t}^k$.

We emphasize that the above formulation does not describe the forward training procedure of federated learning, where the model parameters are optimized with fixed data. Instead, it adopts the perspective of a reconstruction attacker and formulates an inverse problem: given the observed model update $\Delta \theta_t^k$, the attacker treats the private mini-batch data $\{\mathbf{x}_b^k, y_b^k\}_{b=1}^B$ as unknown variables and attempts to find data instances that could have generated the same update through the local training dynamics.
This formulation is commonly used in gradient inversion and data reconstruction attacks, where the goal is to recover training data consistent with a given gradient or model update.
By modeling the private data as optimization variables under the constraint imposed by $\Delta \theta_t^k$, we are able to analyze the identifiability of the reconstruction problem and characterize whether the private data can be uniquely determined from the observed update.
If multiple admissible private batches produce the same observed update, then
the private batch cannot be uniquely determined from $\Delta\theta_t^k$.
The explicit formulation above is used only for intuition; the following
identifiability result is stated directly for the general local-update mapping
$F$ and therefore also applies to the full client-update procedure under the
stated assumptions.

To facilitate the subsequent analysis, we explicitly define the Jacobian matrix of the model update with respect to the mini-batch data.
Let $p$ denote the dimension of each input data sample, and let $\Delta \theta_t^k \in \mathbb{R}^d$ be the model update produced by client $k$ after $E$ local epochs on a mini-batch $\{\mathbf{x}_b^k, y_b^k\}_{b=1}^B$.
The Jacobian matrix $\mathbf{J} \in \mathbb{R}^{d \times Bp}$ is defined as
\[
\mathbf{J} =
\begin{bmatrix}
\frac{\partial \Delta \theta_{t}^k}{\partial \mathbf{x}_{1}^k} &
\frac{\partial \Delta \theta_{t}^k}{\partial \mathbf{x}_{2}^k} &
\cdots &
\frac{\partial \Delta \theta_{t}^k}{\partial \mathbf{x}_{B}^k}
\end{bmatrix},
\]
where each block $\frac{\partial \Delta \theta_{t}^k}{\partial \mathbf{x}_{b}^k} \in \mathbb{R}^{d \times p}$ represents the sensitivity of the model update with respect to the $b$-th input sample.
Accordingly, each row of $\mathbf{J}$ corresponds to one element of $\Delta \theta_t^k$, and each column corresponds to one feature of the batch data.

\begin{thm}\label{thm: not_uniquely_determined}
Let $d$ be the model parameter dimension and $p$ the dimension
of a single input sample. Fix the labels
$\{y_b^k\}_{b=1}^{B}$, the initial model parameters, the local
training configuration, and any auxiliary random seeds.
Write
\[
\mathbf{x}^k
=
\bigl((x_1^k)^\top,\ldots,(x_B^k)^\top\bigr)^\top
\in \mathbb{R}^{Bp}.
\]
Consider the local-update mapping
\[
F:U\longrightarrow\mathbb{R}^{d},
\qquad
F(\mathbf{z})
=
\Delta\theta_t^k
\bigl(\{z_b,y_b^k\}_{b=1}^{B}\bigr),
\]
where $U\subset\mathbb{R}^{Bp}$ is an open neighborhood of
$\mathbf{x}^k$ contained in the admissible input domain.

Assume that $F$ is continuously differentiable on $U$ and that
its Jacobian has constant rank on $U$:
\[
\operatorname{rank} DF(\mathbf{z})
=
r
<
Bp,
\qquad
\mathbf{z}\in U.
\]
Then, on a sufficiently small neighborhood
$V\subset U$ of $\mathbf{x}^k$, the level set
\[
\mathcal{M}
=
\left\{
\mathbf{z}\in V:
F(\mathbf{z})=F(\mathbf{x}^k)
\right\}
\]
is a $C^1$ submanifold of dimension $Bp-r>0$.

Consequently, there exist arbitrarily small nonzero
perturbations
\[
\delta\mathbf{x}^k
=
\bigl((\delta x_1^k)^\top,\ldots,
      (\delta x_B^k)^\top\bigr)^\top
\]
such that $\mathbf{x}^k+\delta\mathbf{x}^k\in U$ and
\[
\Delta\theta_t^k
\bigl(\{x_b^k+\delta x_b^k,y_b^k\}_{b=1}^{B}\bigr)
=
\Delta\theta_t^k
\bigl(\{x_b^k,y_b^k\}_{b=1}^{B}\bigr).
\]
Thus, the private input batch is not locally identifiable
from this model update, even when its labels are fixed.
\end{thm}

Theorem~\ref{thm: not_uniquely_determined} is a local identifiability result under the stated
continuous-differentiability and constant-rank assumptions.
The rank deficiency implies that a positive-dimensional
family of nearby admissible batches produces exactly the
same model update. The null space of the Jacobian describes
the tangent directions of this local level set; an arbitrary
finite step along a null-space vector need not remain on the
level set.

This result does not provide a positive lower bound on the
distance between indistinguishable batches and does not
exclude accurate approximate reconstruction. Moreover, it
concerns identifiability within the stated continuous input
domain from the specified update, rather than an absolute
privacy guarantee against additional observations or side
information.

Since
\begin{equation*}
\Delta \theta_{t}^k 
= -\frac{\eta}{B} \sum_{e=0}^{E-1} \sum_{b=1}^B \nabla_\theta \ell(\theta_{t,e}^k, \mathbf{x}_b^k, y_b^k) \\
= -\frac{\eta E}{B} \sum_{b=1}^B \mathbf{g}_b^k,
\end{equation*}
where $\mathbf{g}_b^k = \frac{1}{E} \sum_{e=0}^{E-1} \nabla_\theta \ell(\theta_{t,e}^k, \mathbf{x}_b^k, y_b^k)$ is the average gradient of data $(\mathbf{x}_b^k, y_b^k)$. 

The Jacobian of the full local-update mapping is
\[
J
=
-\frac{\eta E}{B}
\left[
\frac{\partial\left(\sum_{j=1}^{B}g_j^k\right)}
     {\partial x_1^k}
\quad \cdots \quad
\frac{\partial\left(\sum_{j=1}^{B}g_j^k\right)}
     {\partial x_B^k}
\right]
\in\mathbb{R}^{d\times Bp}.
\]
Each derivative includes both the direct dependence on the
input sample and the indirect dependence through the
intermediate local parameters
$\theta_{t,e}^k=\theta_{t,e}^k(\mathbf{x}^k)$.
In particular, for multiple local steps, changing one
sample can affect the gradient contributions of other
samples through the local parameter trajectory.

Regardless of these dependencies, the matrix dimensions give
\[
\operatorname{rank}(J)\leq \min\{d,Bp\}.
\]
Therefore, if $d<Bp$, the Jacobian cannot have full column
rank. Under the continuous-differentiability and local
constant-rank assumptions, Theorem 4.1 then yields the
following sufficient condition for local non-identifiability. This leads to the following theorem:

\begin{thm}\label{thm: quantitative_dbp}
Consider horizontal federated learning with model parameter
dimension $d$ and input dimension $p$. Fix the labels,
initial model parameters, local training configuration,
and auxiliary random seeds.

Assume that the local-update mapping is continuously
differentiable on an open neighborhood of the true input
batch contained in the admissible input domain, and that
its Jacobian has constant rank on this neighborhood.

If
\[
d<Bp,
\]
then the true input batch is not locally identifiable from
the observed model update. In particular, arbitrarily close
but distinct admissible batches with the same labels
produce exactly the same update.

Thus,
\[
B\geq \left\lfloor\frac{d}{p}\right\rfloor+1
\]
is a sufficient dimension condition for this
non-identifiability conclusion under the stated regularity
assumptions.
\end{thm}

% \begin{proof}
% When $d < Bp$, the rank of the Jacobian matrix $\mathbf{J} \in \mathbb{R}^{d \times Bp}$ is $\mathrm{rank}(\mathbf{J}) \leq d < Bp$. According to \pref{thm: not_uniquely_determined}, the private batch data cannot be uniquely determined from $\Delta \theta_{t}^k$, and privacy is protected.
% \end{proof}

\noindent\textit{Discussion on over-parameterized models.}
Under the stated regularity assumptions, $d<Bp$ is a
sufficient condition for local non-identifiability, not a
necessary condition. Failure of this dimension condition
does not establish unique or accurate reconstructability.
The optimization-theoretic analysis in Section 5 does not
require $d<Bp$, but instead relies on its own gradient,
samplewise residual, and reconstruction-stability assumptions.

Theorem~\ref{thm: quantitative_dbp} concerns a continuous admissible input domain.
It does not, by itself, establish non-identifiability on
a finite token alphabet or another constrained subset
unless distinct feasible alternatives are shown to exist.

Increasing the batch size can make the sufficient dimension
condition easier to satisfy. Separately, under uniform
stability constants and fixed remaining quantities,
Theorem~\ref{thm:privacy_leakage_upper_bound} has a smaller batch-size-dependent concentration
term for larger batches. Neither observation establishes
an unconditional monotonic decrease in actual privacy
leakage. Batch size should therefore be selected together
with the computational budget and training requirements.

A detailed proof is presented in ~\ref{sec: thm: quantitative_dbp}.

\section{Theoretical Analysis from the Perspective of Optimization Theory}\label{sec:average_case_optimization}

In this section, we provide a theoretical analysis of the relationship between privacy leakage and various factors in federated learning, including the number of local data samples, the number of local epochs, and the batch size. We measure the privacy leakage using the discrepancy between the reconstructed data and the original data, as defined in Definition \ref{defi:Privacy_Leakage}. We also introduce the concept of distortion extent in Definition \ref{defi: distortion_extent}, which quantifies the difference between the gradients computed on the original data and the surrogate data at the same model parameters.

For the semi-honest attacker, the privacy leakage is measured using the discrepancy between the reconstructed data and the original data \cite{zhang2023theoretically, zhang2023probably, zhang2024unified}. The privacy leakage is defined as follows.
\begin{definition}[Privacy Leakage]
\label{defi:Privacy_Leakage}
Let $\mathbf{x}_{t,i}^{k}$ represent the $i$-th data sample reconstructed by the attacker at reconstruction round $t$ for client $k$, and $\mathbf{x}_{i}^{k}$ represent the $i$-th original data sample. The privacy leakage is measured as:
\begin{align}
\epsilon_p^{k} = 1 - \mathbb{E}\left[\frac{1}{|\mathcal{D}^{k}|}\sum_{i = 1}^{|\mathcal{D}^{k}|}\frac{1}{T}\sum_{t = 1}^T \frac{||\mathbf{x}_{t,i}^{k} - \mathbf{x}_{i}^{k}||}{D}\right],
\end{align}
where $D$ is a constant satisfying $||\mathbf{x}_{t,i}^{k} - \mathbf{x}_{i}^{k}||\le D$, and the expectation is taken over the randomness in the local dataset $\mathcal{D}^{k}$.
\end{definition}

The batch-level leakage analyzed below provides a tractable empirical approximation of the dataset-level metric through random batch sampling.

\paragraph{Batch-level setting.}
For the batch-level analysis below and in ~\ref{sec: thm:privacy_leakage_upper_bound},
we consider a fixed reconstruction task at a fixed reference
model $\theta$. Here, $T$ denotes the number of attacker
reconstruction iterations, rather than the number of
federated communication rounds used in Sections ~\ref{sec:preliminaries} and ~\ref{sec:linear_equations}.

Write
\[
X=(Z_1,\ldots,Z_B),
\qquad
Z_i=(x_i^k,y_i^k),
\qquad
Z_1,\ldots,Z_B\overset{\mathrm{i.i.d.}}{\sim}P_k,
\]
where $P_k$ is the local sampling distribution.
This setting includes independent sampling with replacement
from a fixed local dataset.

The reference model and all auxiliary random seeds of the
defense and reconstruction procedures are held fixed.
All batch-level probabilities and expectations are taken
over $X$ under this fixed-state sampling model.
In particular, the stated independence is required under
this conditioning.

Let $x_{t,i}^k(X)$ and $\widetilde{x}_i^k(X)$ denote the
reconstructed and surrogate inputs, respectively, using a
consistent sample-matching convention. Their dependence
on the entire batch is retained. Labels are kept unchanged
within each gradient comparison.

The original, surrogate, and reconstructed inputs are
assumed to lie in the input region on which Assumption~\ref{assump: two-sided Lipschitz}
holds, with distances bounded by $D$.
Definition~\ref{defi:Privacy_Leakage} remains a dataset-level expected leakage
metric; it is not identified with a batch-level metric
without an additional consistency argument.
Under the stated i.i.d. sampling model, $\mathbb{E}_X[\hat{\epsilon}^k_{p,B}(X)]$ converges to $\epsilon^k_p$ as $B \to \infty$ by the law of large numbers, so Theorem~\ref{thm:privacy_leakage_upper_bound} bounds a consistent estimator of Definition~\ref{defi:Privacy_Leakage}.

\begin{definition}[Distortion Extent]\label{defi: distortion_extent}
Let $\widetilde{x}_i^k(X)$ denote the surrogate input
associated with the privacy protection mechanism for
the $i$-th sample of batch $X$. The distortion extent
of the current batch is
\begin{equation}
\Delta^k(X)
\coloneqq
\left\|
\frac{1}{B}\sum_{i=1}^{B}
\left[
\nabla_{\theta}\ell
\bigl(\theta,\widetilde{x}_i^k(X),y_i^k\bigr)
-
\nabla_{\theta}\ell
\bigl(\theta,x_i^k,y_i^k\bigr)
\right]
\right\|.
% (6)
\end{equation}
Its expected value is denoted by
\[
\overline{\Delta}^k
\coloneqq
\mathbb{E}_{X}\bigl[\Delta^k(X)\bigr].
\]

We distinguish the realized batch distortion $\Delta^k(X)$
from the deterministic quantity $\overline{\Delta}^k$.
The existence of the surrogate inputs and the applicability
of Assumption 5.1 to them are part of the setting.
The expectation uses the same fixed-state sampling model
as the batch-level leakage analysis.
\end{definition}

\textbf{Remark:} The distortion extent $\Delta^{k}$ quantifies the gradient mismatch introduced by the privacy protection method and does not require an upper bound such as $\Delta^{k} \le 1$. It is a key factor in determining the privacy leakage, influenced by the specific privacy defense employed.

\begin{assumption}[Bi-Lipschitz Condition of Gradients]\label{assump: two-sided Lipschitz}
Let $\|\mathbf{x}_1 - \mathbf{x}_2\|\le D$ for any two data samples $\mathbf{x}_1$ and $\mathbf{x}_2$. We assume that their gradients satisfy the bi-Lipschitz condition with positive constants $c_a$ and $c_b$ \cite{royden1988real} as follows:
\begin{align*}
c_a \|\nabla \ell(\theta, \mathbf{x}_1, \mathbf{y}) - & \nabla \ell(\theta, \mathbf{x}_2, \mathbf{y})\|\le\|\mathbf{x}_1 - \mathbf{x}_2\|\\
&\le c_b \|\nabla \ell(\theta, \mathbf{x}_1, \mathbf{y}) - \nabla \ell(\theta, \mathbf{x}_2, \mathbf{y})\|.
\end{align*}

\textbf{Discussion:}  
Unlike the more common Lipschitz gradient assumption with respect to the model parameters (e.g., $\|\nabla_\theta \ell(\theta_1) - \nabla_\theta \ell(\theta_2)\|\le L\|\theta_1-\theta_2\|$), here we require a two-sided Lipschitz condition in the \emph{input feature space}.
This assumption enables us to relate the distance between reconstructed and true input samples to the distance between their gradients, which is critical in converting reconstructed-data error into gradient-based regret bounds in our privacy analysis.

Such input-based Lipschitz conditions have been considered in the literature on inverse problems and gradient inversion attacks. For example, similar assumptions appear in analyses of gradient inversion and reconstruction, where the gradient of the loss with respect to inputs is assumed to be locally Lipschitz or smooth in the input domain \cite{yin2021see, zhang2023theoretically, zhang2023probably}. In our analysis, the bi-Lipschitz condition is imposed on the local input region relevant to reconstruction and requires not only gradient smoothness but also sufficient local identifiability of the input-to-gradient mapping.

Therefore, this condition should not be interpreted as a universal property of arbitrary neural networks. It may fail when the input-to-gradient mapping is non-injective or degenerate, and our analysis is restricted to regimes where the stated constants $c_a$ and $c_b$ exist.

\end{assumption}

\textbf{Remark:}
The bi-Lipschitz condition is assumed only on the local input region considered in our analysis and is not intended to hold globally for arbitrary neural networks. The right-hand inequality requires local injectivity of the input-to-gradient mapping, which may fail for non-injective or degenerate input-to-gradient mappings, including those that can arise in networks with piecewise-linear activations such as ReLU. Therefore, the following analysis is restricted to regimes where this local identifiability condition holds. Assumption~\ref{assump: two-sided Lipschitz} concerns the single-sample input-to-gradient map
at fixed model parameters and labels.
This does not imply injectivity of the batch-to-update map
in Theorem~\ref{thm: not_uniquely_determined}: distinct batches can produce the same aggregate
update even when the single-sample gradient map is injective.

{\color{black}

\begin{assumption}[Self-bounded Regret]
\label{assump: bounds_for_optimization_alg}
Let $T$ represent the total number of learning rounds for the semi-honest attacker. We assume that the regret of the attacker's optimization process satisfies an upper bound of $O(T^{1/2})$, i.e.,
\begin{equation}
R_T\coloneqq
\sum_{t=1}^{T}\frac{1}{B}\sum_{i=1}^{B}
\left\|
\nabla\ell(\theta,x_{t,i}^k,y_i^k)
-\nabla\ell(\theta,\widetilde{x}_i^k,y_i^k)
\right\|
\leq E c_2\sqrt{T}
\quad\text{almost surely},
\end{equation}
where $c_2>0$ is independent of $T$ and the batch realization,
$x_{t,i}^k$ is the reconstructed sample, and
$\widetilde{x}_i^k$ is the surrogate sample in Definition 5.2.
Here, ``regret'' denotes the displayed samplewise residual sum. The factor \(E\) reflects the dependence of the reconstruction residual on the local training process that generates the observed update. We assume that the accumulated effect of \(E\) local optimization steps can be captured by a linear scaling of the regret coefficient.

\textbf{Discussion:}
This assumption characterizes a worst-case upper bound on the attacker's optimization process. Specifically, we assume that the attacker employs an optimization algorithm whose cumulative gradient mismatch grows no faster than $O(T^{1/2})$. The bound in Eq. (7) is an additional samplewise residual assumption.
We do not require a matching lower bound on the regret, since the actual regret may depend on the difficulty of the reconstruction problem and the specific sequence of examples. For some easy sequences, the regret can grow more slowly than $O(T^{1/2})$. Therefore, the assumption only imposes an upper bound sufficient for our privacy leakage analysis.
\end{assumption}

\textbf{Remark:}
This assumption reflects a realistic scenario where the attacker employs an effective optimization algorithm to minimize the gradient mismatch between reconstructed and target data. Many well-known gradient-based optimizers, such as AdaGrad \cite{duchi2011adaptive} and Adam \cite{kingma2014adam}, provide motivation, but their standard regret guarantees do not establish Eq. (7). The constant $c_2$ captures the dependence of the regret bound on problem-specific factors, such as data dimension and loss function smoothness. This assumption enables us to analyze privacy leakage against a strong attacker equipped with modern optimization techniques while avoiding unnecessarily restrictive requirements on the regret lower bound.

}

% \begin{lem}[Chernoff-Hoeffding Inequality]
% \label{lem:hoeffdingbound}
% Let $X_1, X_2, \dots, X_T$ be i.i.d. random variables supported on $[0, 1]$. For any positive number $\epsilon$, we have:
% \begin{align}
% \Pr\left(\left|\frac{1}{T} \sum_{t = 1}^T X_t - \mathbb{E}\left[\frac{1}{T} \sum_{t = 1}^T X_t\right]\right| \ge \epsilon\right) \le 2 \exp(-2T\epsilon^2).
% \end{align}
% \end{lem}

\begin{lemma}[Concentration of Batch-Average Reconstruction Error]
\label{lem:mcdiarmid_concentration}
Let $x_1^k, \ldots, x_B^k$ be $B$ data samples drawn i.i.d.\ from the local data distribution $\mathcal{D}^k$ of client $k$.
Define the batch-average reconstruction error as
\begin{equation}
    f\!\left(x_1^k, \ldots, x_B^k\right)
    \;=\;
    \frac{1}{B}\sum_{i=1}^{B}\frac{1}{T}\sum_{t=1}^{T}
    \frac{\left\|x_{t,i}^{k} - x_i^k\right\|}{D},
\end{equation}
where $x_{t,i}^{k}$ denotes the attacker's reconstruction of sample $x_i^k$ at iteration $t$, and $D$ is the upper bound on the distance between any two input data samples.
We additionally assume the following reconstruction-stability
condition for every replacement of one labeled record, with labels suppressed in the notation of $f$:
\begin{equation}
    \left|
        f\!\left(x_1^k,\ldots,x_i^k,\ldots,x_B^k\right)
        -
        f\!\left(x_1^k,\ldots,\tilde{x}_i^k,\ldots,x_B^k\right)
    \right|
    \;\leq\;
    \frac{c'}{B},
    \label{eq:bounded_diff}
\end{equation}
where $c'>0$ is a deterministic stability constant uniform
over the batches and batch sizes under consideration.
Both evaluations rerun the full defense and reconstruction
procedures with the same fixed auxiliary seeds; this stability
condition is not inferred from Assumption ~\ref{assump: two-sided Lipschitz}.
Applying this result to a specific defense and reconstruction
procedure requires verifying the stated stability condition.
In particular, comparisons across batch sizes require a common
constant $c'$ that does not grow with $B$; such a constant is not
established here for the evaluated procedures.
Consequently, by McDiarmid's inequality~\cite{mcdiarmid1989},
\begin{equation}
    \Pr\!\left[
        \,f\!\left(x_1^k,\ldots,x_B^k\right)
        \;-\;
        \mathbb{E}\!\left[f\!\left(x_1^k,\ldots,x_B^k\right)\right]
        \;\geq\;
        \varepsilon
    \right]
    \;\leq\;
    \exp\!\left(
        -\,\frac{2B\varepsilon^2}{c'^{\,2}}
    \right).
    \label{eq:mcdiarmid_bound}
\end{equation}
\end{lemma}

To facilitate our analysis, we make two assumptions.
Assumption~\ref{assump: two-sided Lipschitz} imposes a bi-Lipschitz relationship between input-space
distances and gradient-space distances, which allows reconstruction error to
be related to gradient mismatch.
Assumption~\ref{assump: bounds_for_optimization_alg} assumes that the semi-honest attacker's optimization algorithm
has a self-bounded regret, imposed as an additional samplewise residual
condition; the high-probability bound further requires the reconstruction
stability condition in Lemma~\ref{lem:mcdiarmid_concentration}.

Here is a detailed theoretical analysis of the relationship between the upper bound of privacy leakage and the number of local data samples $n$, the number of local epochs $E$, and the batch size $B$ in federated learning. Recent work by \cite{zhang2023theoretically, zhang2023probably, zhang2024unified} has also provided upper bounds on the privacy leakage in federated learning. However, these prior analyses typically consider the scenario where the number of training epochs $E = 1$, and assume that the total number of training samples $n^k$ is equal to the batch size $B$.
The bound concerns the stated fixed-reference, batch-level
reconstruction setting, with $E$ entering through Assumption~5.2
and $n_k$ allowed to differ from $B$.

\begin{thm}[Upper Bound for Privacy Leakage]
\label{thm:privacy_leakage_upper_bound}
Let \pref{assump: two-sided Lipschitz}, \pref{assump: bounds_for_optimization_alg} and the stability condition in Lemma ~\ref{lem:mcdiarmid_concentration} hold under the stated i.i.d. batch sampling model, with $B\geq 2$. Assume that $\overline{\Delta}^k \ge \frac{2 c_2\cdot c_b E}{c_a\sqrt{T}}$, where $c_a, c_2$ and $c_b$ are introduced in \pref{assump: two-sided Lipschitz} and \pref{assump: bounds_for_optimization_alg}. Let $n^{k}$ represent the number of local data samples of client $k$, $E$ represent the number of local epochs, and $B$ represent the batch size. Assume that the assumptions hold, With probability at least $1-\frac{1}{B}$ over batch sampling,
the realized batch leakage $\widehat{\epsilon}_{p,B}^k(X)\coloneqq 1-f(X)$
is bounded by:
\begin{align}
\widehat{\epsilon}_{p,B}^k(X) \le 1 + \sqrt{\frac{c'^{\,2}(\ln 2 + \ln B)}{2B}} - \frac{c_a}{2D} \cdot \overline{\Delta}^k,
\end{align}
where $c'$ is the bounded-difference constant introduced in the concentration analysis, $c_a$ and $D$ are constants defined in the assumptions, and $\overline{\Delta}^k$ is the expected batch distortion.
Since $0 \leq \hat{\epsilon}^k_{p,B}(X) \leq 1$, the stated
upper bound is strictly tighter than the trivial bound $1$ only when
$\overline{\Delta}^k
>
\frac{2D}{c_a}
\sqrt{\frac{c'^2(\ln 2+\ln B)}{2B}}.
$
Otherwise, the bound remains valid but provides no improvement
over the trivial bound.
\end{thm}

This theorem provides insights into the factors that affect the privacy leakage in federated learning. With all remaining quantities fixed, increasing $B$ or $\overline{\Delta}^k$ can reduce this upper bound, but this does not establish monotonicity of actual leakage.
The number of local data samples \(n^{k}\) and the number of local epochs \(E\) do not appear explicitly in the final bound, but they can still affect the privacy analysis through the batch size \(B\), the local training configuration, and the sufficient condition of the theorem. Specifically, when the local data are divided into a fixed number
of equal-sized batches, a larger \(n^{k}\) corresponds to a larger
batch size \(B\), which reduces the concentration term in
Theorem~\ref{thm:privacy_leakage_upper_bound} when the remaining
quantities are fixed. Similarly, the number of local epochs \(E\) appears explicitly in the sufficient condition $\overline{\Delta}^k\ge \frac{2 c_2 c_b E}{c_a\sqrt{T}}$, indicating that the distortion level required for the simplified privacy bound depends on the amount of local training accumulated before communication. A detailed proof is presented in ~\ref{sec: thm:privacy_leakage_upper_bound}.

The significance of this theorem lies in its quantitative characterization of the relationship between privacy leakage and various factors in federated learning. It provides a theoretical foundation for understanding the impact of these factors on privacy and can guide the design of privacy-preserving federated learning algorithms. By carefully tuning the batch size, the number of local data samples, and the number of local epochs, one can potentially achieve a better trade-off between privacy and utility in federated learning.

The upper bound of the privacy leakage $\widehat{\epsilon}_{p,B}^k(X)$ of client $k$ depends on the batch size $B$, the distortion extent $\overline{\Delta}^k$, and several constants. Note that the number of local data samples $n^{k}$ and the number of local epochs $E$ do not directly appear in the final bound. However, they implicitly affect the privacy analysis through the batch size $B$, the training configuration, and the sufficient condition required by the theorem.
Intuitively, the theorem suggests that:
\begin{itemize}
\item Increasing the batch size $B$ can reduce the upper bound of privacy leakage, as the term $\sqrt{\frac{c'^{\,2}(\ln 2 + \ln B)}{2B}}$ decreases asymptotically with larger $B$. In particular, the batch-size-dependent concentration term approaches $0$ as $B$ goes to $\infty$.
\item Increasing the distortion extent $\overline{\Delta}^k$ can also reduce the upper bound of privacy leakage, as the term $\frac{c_a}{2D} \cdot \overline{\Delta}^k$ increases with larger $\overline{\Delta}^k$.
\item The number of local data samples $n^{k}$ affects the privacy leakage through its relationship with the batch size $B$. For example, when the number of local batches is fixed, a larger $n^{k}$ corresponds to a larger batch size, which reduces the batch-size-dependent concentration term in the upper bound.
\item The number of local epochs $E$ affects the privacy analysis through the sufficient condition $\overline{\Delta}^k\ge \frac{2 c_2 c_b E}{c_a\sqrt{T}}$. Under this regret assumption, a larger $E$ increases the distortion extent required for the simplified privacy leakage bound to hold, reflecting the effect of accumulating more local training steps before communication.
\end{itemize}

It's worth noting that the actual impact of $n^{k}$ and $E$ on privacy leakage may vary depending on the specific attack methods, federated learning configurations, and the assumptions made. The provided theorem gives a general upper bound based on the stated assumptions, but the relationship between privacy leakage and these factors can be complex and requires further analysis in specific scenarios.

\section{Experimental Evaluation of Batch Size Impact on Gradient Leakage and Performance}
\label{experiment}

\subsection{Experimental Setups}

\paragraph{FL Settings}
By default, the federated learning system consists of 4 clients. The global training process runs for 30 communication rounds, with each client performing 1 local training epoch per round. All datasets are split into 70\% training, 15\% validation, and 15\% testing, with data equally partitioned among clients. For MNIST and FashionMNIST, we adopt the LeNet architecture, and for CIFAR-10 we use the ConvNet-64 model. 
For the more complex datasets CIFAR-100 and Tiny-ImageNet, we adopt ResNet-50 as the backbone model to better capture high-dimensional and semantic-rich visual features. 
Both local model updates and perturbation generation are optimized using SGD with a learning rate of 0.1 and no weight decay. Gradient leakage attacks are simulated using the Inverting Gradients method~\cite{geiping2020inverting}. 
To emulate the distortion effect assumed in Theorem~5.1, we add a random perturbation sampled uniformly from an $\ell_2$-ball of radius $1/255$ to the computed gradients before transmission. This introduces a bounded, stochastic gradient distortion whose scale is controlled by the sampling radius, which is fixed across all experiments.

\paragraph{Datasets and Evaluation Metrics}
We conduct experiments on five widely used benchmark datasets in federated learning: MNIST~\cite{lecun1998gradient}, FashionMNIST~\cite{xiao2017fashion}, and CIFAR-10~\cite{krizhevsky2009learning}, CIFAR-100~\cite{krizhevsky2009learning} and Tiny-ImageNet. 
We evaluate model utility using validation and test accuracy. To assess privacy protection, we measure the quality of reconstructed images obtained by attackers using metrics such as Mean Squared Error (MSE), Structural Similarity Index Measure (SSIM)\cite{wang2004image}, Peak Signal-to-Noise Ratio (PSNR), Learned Perceptual Image Patch Similarity (LPIPS)\cite{zhang2018unreasonable}, and Kullback-Leibler (KL) divergence, which together quantify the difference between reconstructed and original samples.

\paragraph{Privacy Metric Computation}To quantitatively evaluate privacy leakage from gradient inversion, we employ five commonly used similarity metrics between the reconstructed image $\hat{x}$ and the original image $x$: MSE, SSIM~\cite{wang2004image}, PSNR, LPIPS~\cite{zhang2018unreasonable}, and KL divergence.

MSE measures the average pixel-wise squared error between two images and is computed as:
\[
\text{MSE}(x, \hat{x}) = \frac{1}{n} \sum_{i=1}^{n} (x_i - \hat{x}_i)^2,
\]
where $n$ is the total number of pixels.

SSIM compares two images in terms of luminance, contrast, and structure. It is computed using local image statistics:
\[
\text{SSIM}(x, \hat{x}) = \frac{(2\mu_x \mu_{\hat{x}} + C_1)(2\sigma_{x\hat{x}} + C_2)}{(\mu_x^2 + \mu_{\hat{x}}^2 + C_1)(\sigma_x^2 + \sigma_{\hat{x}}^2 + C_2)},
\]
where $\mu$ and $\sigma$ denote mean and standard deviation of local patches, and $C_1$, $C_2$ are small constants to stabilize the division.

PSNR evaluates image reconstruction quality using the MSE and is defined as:
\[
\text{PSNR}(x, \hat{x}) = 10 \cdot \log_{10} \left( \frac{L^2}{\text{MSE}(x, \hat{x})} \right),
\]
where $L$ is the maximum possible pixel value (e.g., 1.0 or 255 depending on normalization).

LPIPS is a learned perceptual metric that compares feature activations from a deep neural network. We use a pretrained VGG-16 model to extract features from multiple layers and compute weighted $\ell_2$ distances:
\[
\text{LPIPS}(x, \hat{x}) = \sum_{l} \frac{1}{H_l W_l} \sum_{h, w} \| w_l \odot (\phi_l(x)_{hw} - \phi_l(\hat{x})_{hw}) \|_2^2,
\]
where $\phi_l(\cdot)$ denotes the $l$-th layer's feature map, $w_l$ is a learned channel-wise weight, and $(h, w)$ indexes spatial positions.

KL divergence is used to assess semantic-level leakage by comparing the predicted label distributions of $x$ and $\hat{x}$. After passing both images through a pretrained VGG-16 classifier with softmax output, the divergence is computed as:
\[
\text{KL}(P \| \hat{P}) = \sum_{i=1}^{C} P_i \log \left( \frac{P_i}{\hat{P}_i} \right),
\]
where $P$ and $\hat{P}$ are the output probability distributions over $C$ classes.

\subsection{Results}

\begin{table*}[htbp]
\caption{
Performance of FedSGD under different batchsize on the MNIST dataset. 
Metrics are grouped into utility (\textbf{U}) and privacy (\textbf{P}) categories. Gradient leakage attacks are launched at epochs E1. Arrows indicate desired direction: $\uparrow$ = higher is better, $\downarrow$ = lower is better.
}
\label{tab:mnist}
\begin{center}
\begin{small}
\resizebox{\textwidth}{!}{
\begin{tabular}{lcccccccc}
\toprule
Batchsize & Val Acc~\textbf{(U$\uparrow$)} 
& Test Acc~\textbf{(U$\uparrow$)} 
& Stage 
& Test MSE~\textbf{(P$\uparrow$)} 
& PSNR~\textbf{(P$\downarrow$)} 
& SSIM~\textbf{(P$\downarrow$)} 
& LPIPS~\textbf{(P$\uparrow$)} 
& KL~\textbf{(P$\uparrow$)} \\
\midrule
8  & $0.9805 \pm 0.0022$ & $0.9751 \pm 0.0027$ & E1 & $1.2935 \pm 0.0214$ & $9.2315 \pm 0.0842$ & $0.0985 \pm 0.0042$ & $0.5888 \pm 0.0124$ & $3.6115 \pm 0.0452$\\

16  & $0.9772 \pm 0.0018$ & $0.9723 \pm 0.0025$ & E1 & $1.0195 \pm 0.0185$ & $10.451 \pm 0.0765$ & $0.1115 \pm 0.0038$ & $0.6294 \pm 0.0105$ & $2.8761 \pm 0.0384$ \\
        
32  & $0.9791 \pm 0.0015$ & $0.9719 \pm 0.0021$ & E1 & $1.1055 \pm 0.0152$ & $9.9382 \pm 0.0654$ & $0.0871 \pm 0.0031$ & $0.6438 \pm 0.0092$ & $3.4585 \pm 0.0321$ \\
      
64 & $0.9748 \pm 0.0012$ & $0.9692 \pm 0.0018$ & E1 & $1.2405 \pm 0.0121$ & $9.4025 \pm 0.0512$ & $0.0605 \pm 0.0025$ & $0.6795 \pm 0.0075$ & $3.6468 \pm 0.0285$ \\

128 & $0.9701 \pm 0.0010$ & $0.9615 \pm 0.0015$ & E1 & $1.2815 \pm 0.0095$ & $9.2652 \pm 0.0423$ & $0.0535 \pm 0.0018$ & $0.6792 \pm 0.0061$ & $3.6515 \pm 0.0214$ \\
        
\bottomrule
\end{tabular}
}
\end{small}
\end{center}
\vskip -0.1in
\end{table*}

We evaluate the impact of batch size on both utility and privacy leakage under FedSGD using MNIST, FMNIST and CIFAR-10 datasets. As shown in Tables~\ref{tab:mnist}, \ref{tab:fmnist}, and~\ref{tab:cifar10}, the effect of batch size on model utility varies across datasets. Among the tested batch sizes, the highest test accuracy is achieved at $B=8$ for MNIST, $B=32$ for FMNIST, and $B=16$ for CIFAR-10.

\begin{table*}[htbp]
\caption{
Performance of FedSGD under different batchsize on the FMNIST dataset. 
Metrics are grouped into utility (\textbf{U}) and privacy (\textbf{P}) categories. Gradient leakage attacks are launched at epochs E1. Arrows indicate desired direction: $\uparrow$ = higher is better, $\downarrow$ = lower is better.
}
\label{tab:fmnist}
\begin{center}
\begin{small}
\resizebox{\textwidth}{!}{
\begin{tabular}{lcccccccc}
\toprule
Batchsize & Val Acc~\textbf{(U$\uparrow$)} 
& Test Acc~\textbf{(U$\uparrow$)} 
& Stage 
& Test MSE~\textbf{(P$\uparrow$)} 
& PSNR~\textbf{(P$\downarrow$)} 
& SSIM~\textbf{(P$\downarrow$)} 
& LPIPS~\textbf{(P$\uparrow$)} 
& KL~\textbf{(P$\uparrow$)} \\
\midrule
8  & $0.8562 \pm 0.0042$ & $0.8431 \pm 0.0051$ & E1 & $1.4018 \pm 0.0312$ & $8.1275 \pm 0.184$ & $0.1472 \pm 0.012$ & $0.5521 \pm 0.024$ & $2.1035 \pm 0.082$\\

16  & $0.8654 \pm 0.0035$ & $0.8535 \pm 0.0048$ & E1 & $0.9795 \pm 0.0245$ & $9.4218 \pm 0.156$ & $0.2255 \pm 0.015$ & $0.5005 \pm 0.018$ & $2.9648 \pm 0.095$ \\
        
32  & $0.8741 \pm 0.0058$ & $0.8631 \pm 0.0062$ & E1 & $1.0872 \pm 0.0412$ & $8.9392 \pm 0.213$ & $0.1932 \pm 0.018$ & $0.4938 \pm 0.021$ & $2.1225 \pm 0.114$ \\
      
64 & $0.8674 \pm 0.0028$ & $0.8565 \pm 0.0031$ & E1 & $1.1055 \pm 0.0185$ & $8.8122 \pm 0.095$ & $0.1305 \pm 0.008$ & $0.5845 \pm 0.012$ & $2.0441 \pm 0.063$ \\

128 & $0.8621 \pm 0.0045$ & $0.8428 \pm 0.0054$ & E1 & $0.9651 \pm 0.0352$ & $9.5112 \pm 0.162$ & $0.1535 \pm 0.014$ & $0.5621 \pm 0.019$ & $2.2931 \pm 0.088$ \\
        
\bottomrule
\end{tabular}
}
\end{small}
\end{center}
\vskip -0.1in
\end{table*}

For privacy leakage, the results generally align with the theoretical bound in Theorem~\ref{thm:privacy_leakage_upper_bound}, whose concentration term decreases with batch size when the
remaining quantities are fixed. In both MNIST and FMNIST, this trend holds for batch sizes greater than 8: as batch size increases, PSNR and SSIM steadily decline, while LPIPS tends to rise, indicating reduced visual and semantic similarity in reconstructed images. Notably, CIFAR-10 exhibits an even more consistent and monotonic pattern across all batch sizes, with privacy metrics improving continuously from batch size 8 to 128. This highlights the stronger alignment of CIFAR-10 results with our theoretical expectations, suggesting that the privacy benefits of larger batch sizes generalize well to more complex, high-dimensional datasets.

On MNIST and FMNIST, \(B=8\) yields the highest MSE and lowest PSNR among the tested batch sizes. This observation may reflect differences in batch composition and their influence on the reported reconstruction metrics. Reconstruction difficulty varies across samples, and a few poorly reconstructed samples can have a substantial influence on the average error of a small batch. Such composition effects can produce local non-monotonicity in empirical reconstruction quality, without contradicting the conditional batch-size dependence of the theoretical leakage upper bound. This provides a plausible explanation for the \(B=8\) observation without suggesting that smaller batches inherently offer stronger privacy protection.

\begin{table*}[htbp]
\caption{
Performance of FedSGD under different batchsize on the CIFAR-10 dataset. 
Metrics are grouped into utility (\textbf{U}) and privacy (\textbf{P}) categories. Gradient leakage attacks are launched at epochs E1. Arrows indicate desired direction: $\uparrow$ = higher is better, $\downarrow$ = lower is better.
}
\label{tab:cifar10}
\begin{center}
\begin{small}
\resizebox{\textwidth}{!}{
\begin{tabular}{lcccccccc}
\toprule
Batchsize & Val Acc~\textbf{(U$\uparrow$)} 
& Test Acc~\textbf{(U$\uparrow$)} 
& Stage 
& Test MSE~\textbf{(P$\uparrow$)} 
& PSNR~\textbf{(P$\downarrow$)} 
& SSIM~\textbf{(P$\downarrow$)} 
& LPIPS~\textbf{(P$\uparrow$)} 
& KL~\textbf{(P$\uparrow$)} \\
\midrule
8  & $0.6713 \pm 0.0054$ & $0.6842 \pm 0.0068$ & E1 & $1.8398 \pm 0.0421$ & $9.5305 \pm 0.213$ & $0.0241 \pm 0.003$ & $0.6822 \pm 0.015$ & $3.8315 \pm 0.104$\\

16  & $0.7045 \pm 0.0031$ & $0.7155 \pm 0.0039$ & E1 & $2.0312 \pm 0.0285$ & $9.0335 \pm 0.142$ & $0.0265 \pm 0.002$ & $0.7052 \pm 0.011$ & $3.5421 \pm 0.087$ \\
        
32  & $0.6771 \pm 0.0042$ & $0.6848 \pm 0.0045$ & E1 & $1.9692 \pm 0.0354$ & $9.2481 \pm 0.165$ & $0.0304 \pm 0.004$ & $0.7312 \pm 0.014$ & $3.4412 \pm 0.092$ \\
      
64 & $0.6882 \pm 0.0037$ & $0.6985 \pm 0.0041$ & E1 & $1.9885 \pm 0.0312$ & $9.1955 \pm 0.151$ & $0.1552 \pm 0.008$ & $0.7428 \pm 0.012$ & $3.0292 \pm 0.076$ \\

128 & $0.6125 \pm 0.0061$ & $0.6155 \pm 0.0072$ & E1 & $2.0912 \pm 0.0514$ & $9.0182 \pm 0.194$ & $0.0382 \pm 0.005$ & $0.7108 \pm 0.018$ & $3.1105 \pm 0.112$ \\
        
\bottomrule
\end{tabular}
}
\end{small}
\end{center}
\vskip -0.1in
\end{table*}

\begin{figure*}[htbp]
\centering
\begin{minipage}{0.32\textwidth}
    \centering
    \includegraphics[width=\linewidth]{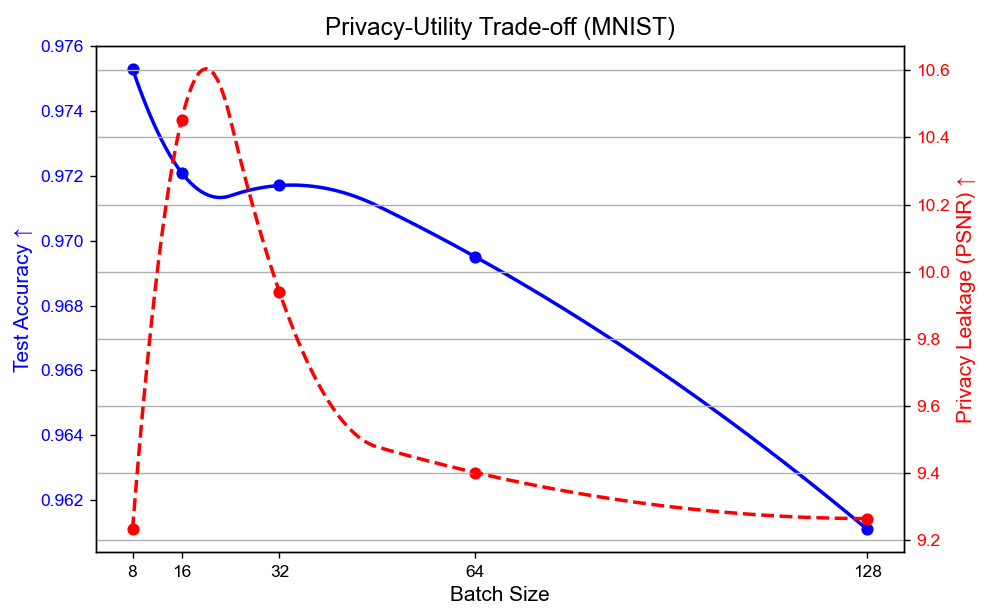}
    \subcaption{MNIST}
\end{minipage}
\hfill
\begin{minipage}{0.32\textwidth}
    \centering
    \includegraphics[width=\linewidth]{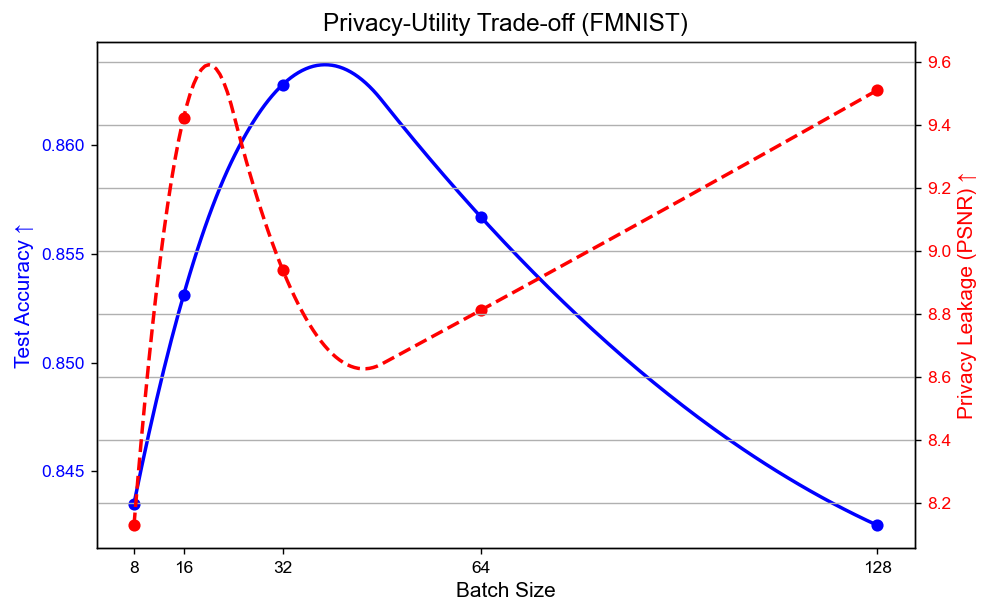}
    \subcaption{FMNIST}
\end{minipage}
\hfill
\begin{minipage}{0.32\textwidth}
    \centering
    \includegraphics[width=\linewidth]{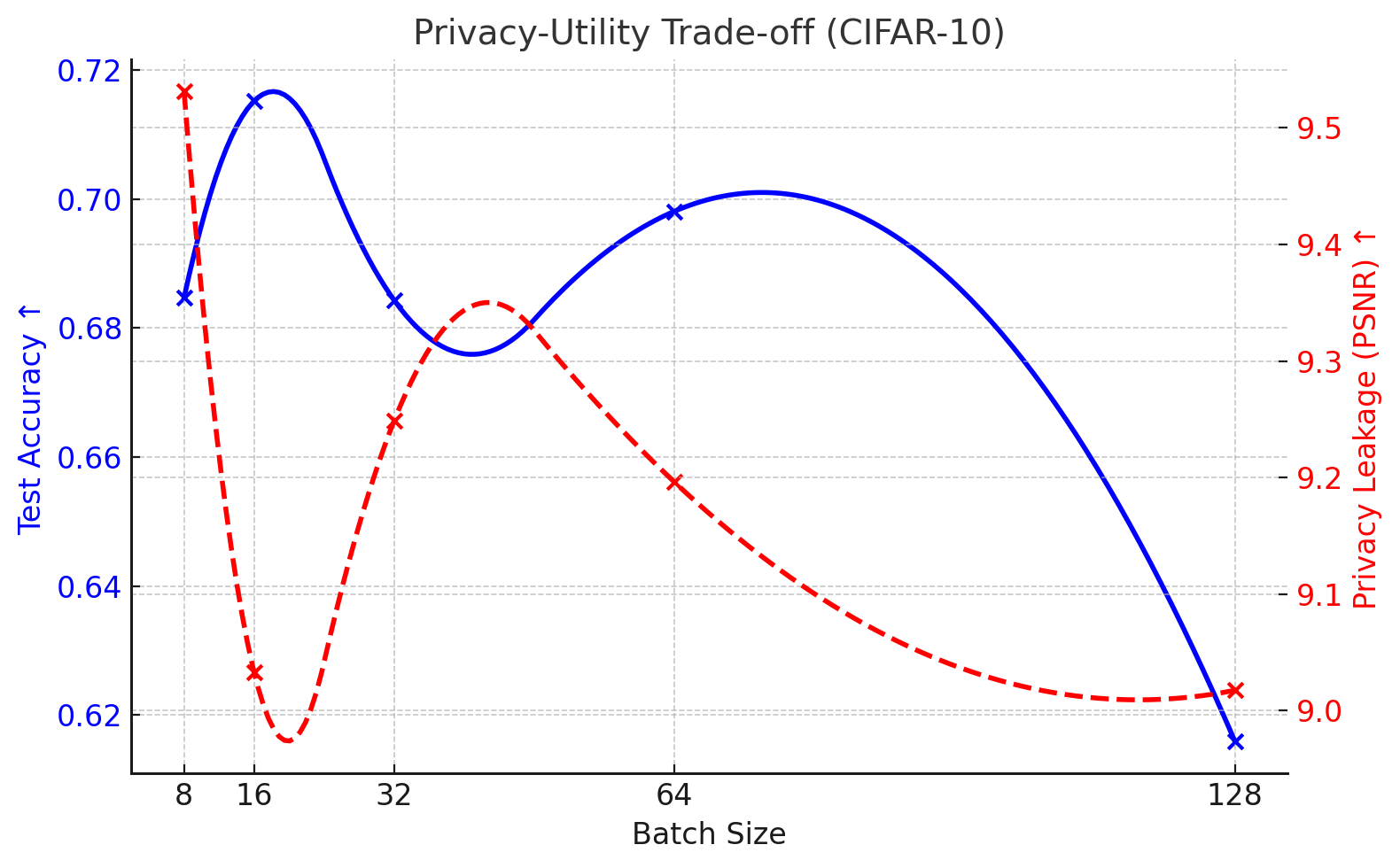} 
    \subcaption{CIFAR-10}
\end{minipage}
\caption{Privacy-utility trade-off curves under different batch sizes. Blue line: test accuracy; red dashed line: privacy leakage measured by PSNR.}
\label{fig:tradeoff}
\end{figure*}

However, FMNIST displays a slight anomaly where large batch sizes (e.g., 128) still yield relatively high PSNR, suggesting a more effective reconstruction. This might be due to the higher structural redundancy or alignment in FMNIST features, making it easier to invert even with aggregated gradients.

\begin{table*}[htbp]
\caption{
Performance of FedSGD under different batchsize on the CIFAR-100 dataset. 
Metrics are grouped into utility (\textbf{U}) and privacy (\textbf{P}) categories. Gradient leakage attacks are launched at epochs E1. Arrows indicate desired direction: $\uparrow$ = higher is better, $\downarrow$ = lower is better.
}
\label{tab:cifar100}
\begin{center}
\begin{small}
\resizebox{\textwidth}{!}{
\begin{tabular}{lcccccccc}
\toprule
Batchsize & Val Acc~\textbf{(U$\uparrow$)} 
& Test Acc~\textbf{(U$\uparrow$)} 
& Stage 
& Test MSE~\textbf{(P$\uparrow$)} 
& PSNR~\textbf{(P$\downarrow$)} 
& SSIM~\textbf{(P$\downarrow$)} 
& LPIPS~\textbf{(P$\uparrow$)} 
& KL~\textbf{(P$\uparrow$)} \\
\midrule
8  & $0.2562 \pm 0.0084$ & $0.2471 \pm 0.0092$ & E1 & $2.3295 \pm 0.0712$ & $7.9941 \pm 0.284$ & $0.0485 \pm 0.005$ & $0.7002 \pm 0.024$ & $3.0878 \pm 0.142$\\

16  & $0.2738 \pm 0.0065$ & $0.2567 \pm 0.0071$ & E1 & $2.1285 \pm 0.0452$ & $8.4455 \pm 0.192$ & $0.0421 \pm 0.004$ & $0.7121 \pm 0.018$ & $3.2805 \pm 0.115$ \\
        
32  & $0.3075 \pm 0.0052$ & $0.2963 \pm 0.0068$ & E1 & $2.1705 \pm 0.0518$ & $8.2872 \pm 0.215$ & $0.0435 \pm 0.004$ & $0.6789 \pm 0.021$ & $3.5732 \pm 0.108$ \\
      
64 & $0.2995 \pm 0.0048$ & $0.2915 \pm 0.0055$ & E1 & $2.1768 \pm 0.0384$ & $8.3331 \pm 0.156$ & $0.0485 \pm 0.003$ & $0.7163 \pm 0.015$ & $2.8715 \pm 0.094$ \\

128 & $0.2718 \pm 0.0057$ & $0.2608 \pm 0.0062$ & E1 & $2.9821 \pm 0.0625$ & $6.8322 \pm 0.231$ & $0.0322 \pm 0.005$ & $0.7405 \pm 0.019$ & $3.5442 \pm 0.126$ \\
        
\bottomrule
\end{tabular}
}
\end{small}
\end{center}
\vskip -0.1in
\end{table*}

Table~\ref{tab:cifar100} reports the impact of batch size on both model utility and privacy leakage for the CIFAR-100 dataset.
In terms of utility, validation and test accuracy exhibit a clear non-monotonic trend: performance improves as the batch size increases from small values to a moderate range, and then degrades when the batch size becomes large. This behavior is consistent with the common generalization trade-off observed in deep learning, where excessively large batches may harm convergence and generalization.

From a privacy perspective, the CIFAR-100 results generally
show reduced reconstruction quality at larger batch sizes. Except for the smallest batch size, which exhibits an anomalous behavior likely caused by increased stochasticity in small-batch gradients, larger batch sizes generally lead to weaker gradient leakage. Across most batch sizes, increasing the batch size consistently degrades reconstruction quality, indicating stronger privacy protection. This trend is qualitatively consistent with the decreasing
concentration term in
Theorem~\ref{thm:privacy_leakage_upper_bound}
when the remaining quantities are fixed.

\begin{table*}[htbp]
\caption{
Performance of FedSGD under different batchsize on the Tiny-imagenet dataset. 
Metrics are grouped into utility (\textbf{U}) and privacy (\textbf{P}) categories. Gradient leakage attacks are launched at epochs E1. Arrows indicate desired direction: $\uparrow$ = higher is better, $\downarrow$ = lower is better.
}
\label{tab:imagenet}
\begin{center}
\begin{small}
\resizebox{\textwidth}{!}{
\begin{tabular}{lcccccccc}
\toprule
Batchsize & Val Acc~\textbf{(U$\uparrow$)} 
& Test Acc~\textbf{(U$\uparrow$)} 
& Stage 
& Test MSE~\textbf{(P$\uparrow$)} 
& PSNR~\textbf{(P$\downarrow$)} 
& SSIM~\textbf{(P$\downarrow$)} 
& LPIPS~\textbf{(P$\uparrow$)} 
& KL~\textbf{(P$\uparrow$)} \\
\midrule
8  & $0.1315 \pm 0.0072$ & $0.1382 \pm 0.0084$ & E1 & $2.8398 \pm 0.0642$ & $6.7412 \pm 0.254$ & $0.0255 \pm 0.004$ & $0.7451 \pm 0.021$ & $5.8205 \pm 0.184$\\

16  & $0.1564 \pm 0.0051$ & $0.1592 \pm 0.0058$ & E1 & $1.9692 \pm 0.0415$ & $8.3065 \pm 0.187$ & $0.0325 \pm 0.003$ & $0.7235 \pm 0.018$ & $5.8421 \pm 0.142$ \\
        
32  & $0.1815 \pm 0.0065$ & $0.1853 \pm 0.0071$ & E1 & $2.1505 \pm 0.0528$ & $7.9851 \pm 0.211$ & $0.0281 \pm 0.005$ & $0.7638 \pm 0.024$ & $4.5872 \pm 0.165$ \\
      
64 & $0.1548 \pm 0.0048$ & $0.1647 \pm 0.0052$ & E1 & $2.1715 \pm 0.0384$ & $7.9352 \pm 0.156$ & $0.0185 \pm 0.002$ & $0.7441 \pm 0.015$ & $5.5871 \pm 0.128$ \\

128 & $0.1128 \pm 0.0081$ & $0.1242 \pm 0.0093$ & E1 & $2.1732 \pm 0.0712$ & $7.9515 \pm 0.234$ & $0.0244 \pm 0.006$ & $0.7642 \pm 0.028$ & $4.9108 \pm 0.201$ \\
        
\bottomrule
\end{tabular}
}
\end{small}
\end{center}
\vskip -0.1in
\end{table*}

Table~\ref{tab:imagenet} summarizes the effect of batch size on utility and privacy leakage for Tiny-ImageNet.
In terms of utility, model performance improves as the batch size increases from small values to a moderate range, and then degrades for larger batch sizes, reflecting the increased optimization difficulty and generalization challenges on this complex dataset.

Regarding privacy, the Tiny-ImageNet results show that reconstruction quality varies with batch size. PSNR decreases as the batch size increases from $16$ to $64$, followed by a slight increase at $128$, indicating that the observed trend is not strictly monotonic.

We visualize the privacy-utility trade-off in Figure~\ref{fig:tradeoff}, where test accuracy and PSNR are plotted on dual axes. Across all three datasets—MNIST, FMNIST, and CIFAR-10—the plots confirm that while moderate batch sizes (e.g., 32) often provide the best utility, privacy risk generally decreases as batch size increases. This consistent trend highlights the need to carefully balance performance and privacy through informed batch size selection.

% =========================
% 1) CIFAR-100: 3-in-1 figure
% =========================
\begin{figure*}[htbp]
\centering
\begin{subfigure}[t]{0.32\textwidth}
    \centering
    \includegraphics[width=\linewidth]{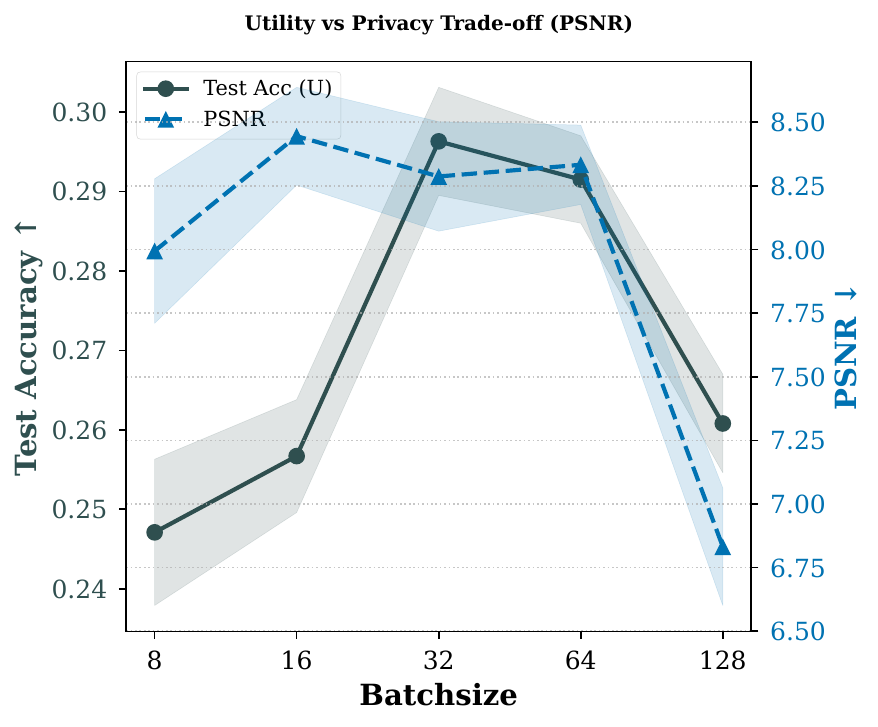}
    \subcaption{CIFAR-100: Utility vs.\ PSNR}
\end{subfigure}
\hfill
\begin{subfigure}[t]{0.32\textwidth}
    \centering
    \includegraphics[width=\linewidth]{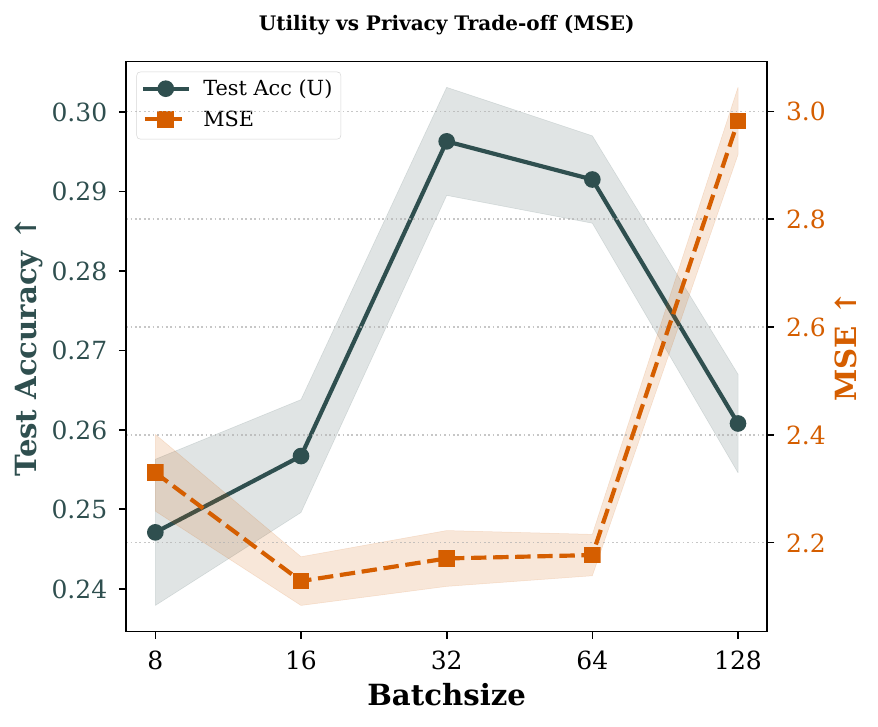}
    \subcaption{CIFAR-100: Utility vs.\ MSE}
\end{subfigure}
\hfill
\begin{subfigure}[t]{0.32\textwidth}
    \centering
    \includegraphics[width=\linewidth]{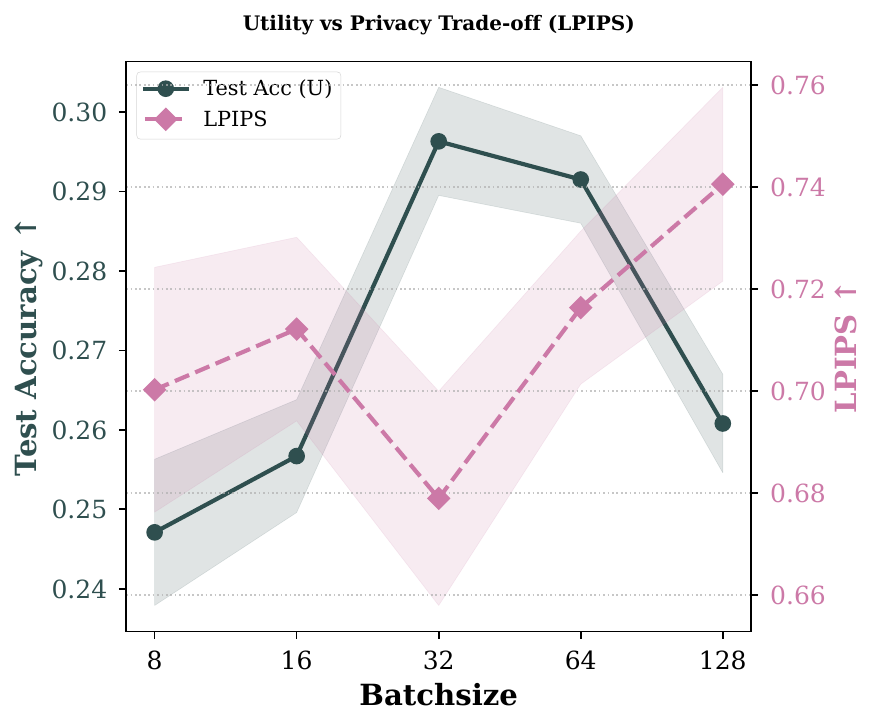}
    \subcaption{CIFAR-100: Utility vs.\ LPIPS}
\end{subfigure}
\caption{CIFAR-100: Privacy--utility trade-off under different batch sizes. Test accuracy is shown as the utility metric, while PSNR/MSE/LPIPS quantify reconstruction quality (privacy leakage).}
\label{fig:tradeoff_cifar100}
\end{figure*}

As shown in Figures~\ref{fig:tradeoff_cifar100} and~\ref{fig:tradeoff_tinyimagenet}, we visualize the privacy--utility trade-off on the CIFAR-100 and Tiny-ImageNet datasets by jointly plotting test accuracy against multiple reconstruction-based privacy metrics, including PSNR, MSE, and LPIPS, under different batch sizes.
Across both datasets, test accuracy typically improves as the batch size increases from small to moderate values and then degrades for larger batch sizes, while privacy leakage exhibits an opposite trend. For example, considering PSNR as a representative metric, larger batch sizes consistently lead to lower PSNR values, indicating poorer reconstruction quality and thus stronger privacy protection. Similar trends are observed for MSE and LPIPS, collectively highlighting a clear trade-off between model utility and resistance to gradient inversion attacks.

\begingroup
\color{black}

% =========================
% 2) Tiny-ImageNet: 3-in-1 figure
% =========================
\begin{figure*}[htbp]
\centering
\begin{subfigure}[t]{0.32\textwidth}
    \centering
    \includegraphics[width=\linewidth]{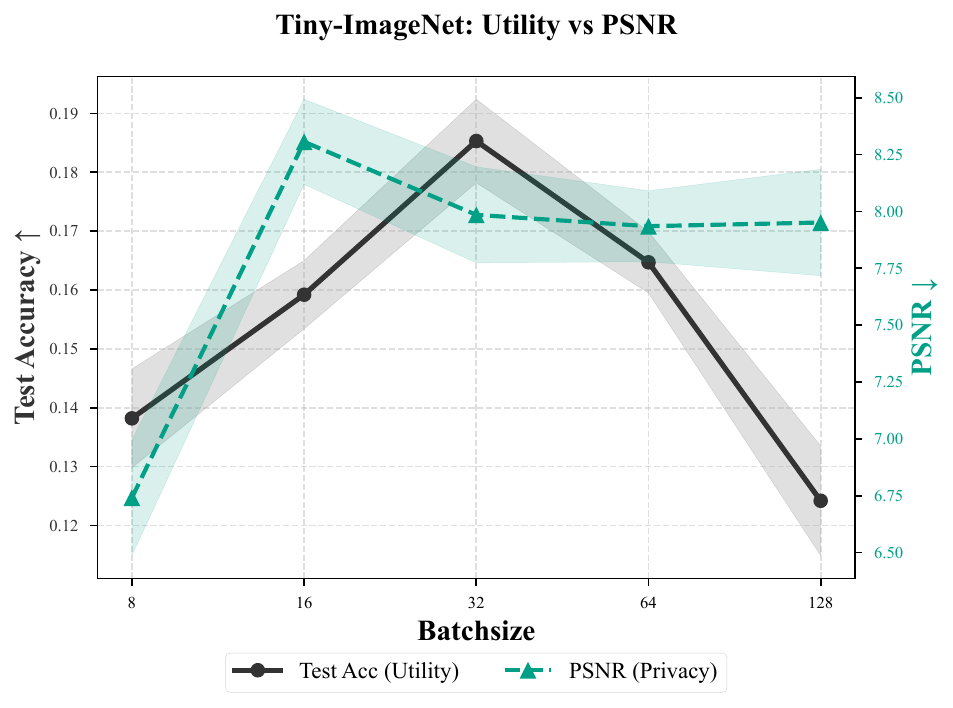}
    \subcaption{Tiny-ImageNet: Utility vs.\ PSNR}
\end{subfigure}
\hfill
\begin{subfigure}[t]{0.32\textwidth}
    \centering
    \includegraphics[width=\linewidth]{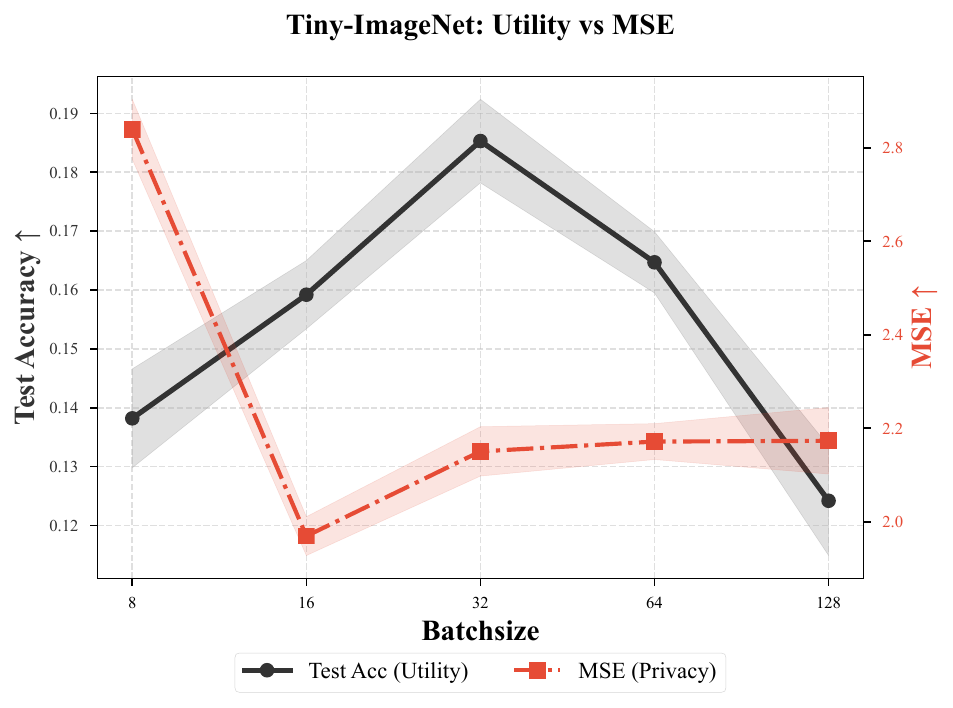}
    \subcaption{Tiny-ImageNet: Utility vs.\ MSE}
\end{subfigure}
\hfill
\begin{subfigure}[t]{0.32\textwidth}
    \centering
    \includegraphics[width=\linewidth]{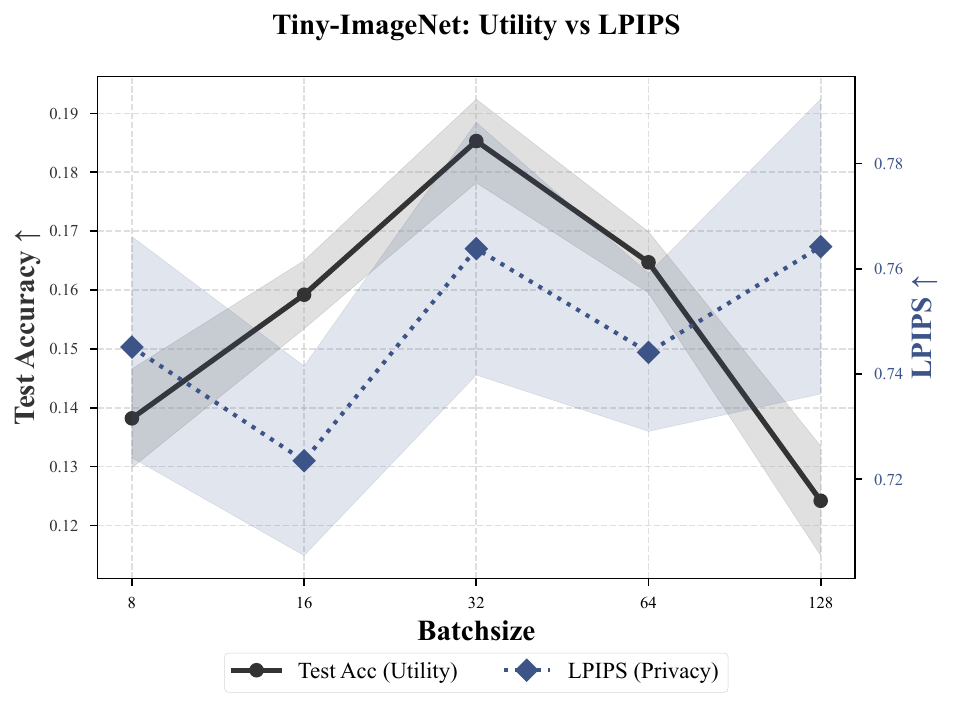}
    \subcaption{Tiny-ImageNet: Utility vs.\ LPIPS}
\end{subfigure}
\caption{Tiny-ImageNet: Privacy--utility trade-off under different batch sizes. The overall trend matches CIFAR-100: moderate batch sizes often yield better utility, while larger batch sizes generally lead to weaker gradient leakage.}
\label{fig:tradeoff_tinyimagenet}
\end{figure*}

\subsection{Evaluation on Modern Language and Vision Models}
\label{sec:llm_validation}

To further examine the empirical relevance of our theoretical analysis to modern language-model settings, we conduct additional experiments covering several representative reconstruction attacks and model architectures.
Specifically, we consider DAGER for analytical gradient inversion,
PEFTLeak for parameter-efficient update leakage, ReCIT for PII-oriented
reconstruction, and partial-gradient reconstruction for leakage from
intermediate Transformer modules.
Our evaluation spans BERT, GPT-2, Qwen3-1.7B-Base, LLaMA-2 7B,
Bloomz-3B, and ViT-B/16, and includes full-parameter training,
parameter-efficient fine-tuning, and partial-gradient exposure settings.
Across these complementary attack scenarios, we empirically examine
how batch size affects reconstruction quality and downstream model
utility. Comparisons with our theoretical analysis are qualitative;
these experiments do not establish that its assumptions hold or that
actual leakage must decrease with batch size.

% ============================================================
\subsubsection{Evaluation under DAGER}
\label{sec:modern_dager}
% ============================================================

\paragraph{Attack Overview}
We first evaluate DAGER~\cite{petrov2024dager}, a recent analytical
gradient inversion attack for Transformer-based language models.
DAGER exploits the low-rank structure of self-attention gradients and
the discrete token space to recover private text, making it particularly
relevant to our analysis of batch size, model dimensionality, and
gradient structure.

We evaluate BERT$_{\mathrm{BASE}}$, GPT-2$_{\mathrm{BASE}}$,
GPT-2$_{\mathrm{LARGE}}$, LLaMA-2 7B, and Qwen3-1.7B-Base over batch
sizes from $B=1$ to $B=128$.
Reconstruction is primarily evaluated using ROUGE-1 (R-1) and
ROUGE-2 (R-2), while downstream utility is measured using MCC on CoLA
and classification accuracy on SST-2 and Rotten Tomatoes (RT).

% ============================================================
% Qwen3
% ============================================================

\paragraph{Reconstruction-Relevant Candidate Information on Qwen3}
We first evaluate Qwen3-1.7B-Base on SST-2 with
$B\in\{1,2,4,8\}$.
Since DAGER identifies candidate tokens before sequence reconstruction,
we measure the Layer-1 candidate-set coverage
\begin{equation}
C_{\mathrm{L1}}
=
\frac{1}{N}
\sum_{i=1}^{N}
\frac{
|\mathcal{T}_i\cap\mathcal{C}^{(1)}_i|
}{
|\mathcal{T}_i|
}
\times100\%,
\end{equation}
where $\mathcal{T}_i$ and $\mathcal{C}^{(1)}_i$ denote the ground-truth
and Layer-1 candidate token sets, respectively.
Lower coverage and more empty candidate sets indicate less information
available for subsequent reconstruction.

\begin{table}[t]
\centering
\small
\begin{tabular}{@{}lrrr@{}}
\toprule
Batch Size ($B$)
& Accuracy (\%) $\uparrow$
& $C_{\mathrm{L1}}$ (\%) $\downarrow$
& Empty Sets (/60) $\uparrow$ \\
\midrule
1 & $\mathbf{94.15}$ & 92.28 & 0 \\
2 & 94.12 & 58.45 & 4 \\
4 & 94.08 & 18.12 & 27 \\
8 & 93.98 & $\mathbf{1.05}$ & $\mathbf{58}$ \\
\bottomrule
\end{tabular}
\caption{
\textcolor{black}{
Qwen3-1.7B-Base Layer-1 utility and reconstruction-relevant
candidate information under different batch sizes on SST-2 without
defense.
Lower candidate-set coverage and more empty candidate sets indicate
less information available for subsequent DAGER reconstruction.
}}
\label{tab:qwen3-batch-size-layer1}
\end{table}

As shown in Table~\ref{tab:qwen3-batch-size-layer1},
$C_{\mathrm{L1}}$ decreases sharply from $92.28\%$ at $B=1$ to
$1.05\%$ at $B=8$, while empty candidate sets increase from 0 to 58
out of 60 cases.
Meanwhile, SST-2 accuracy remains nearly unchanged
($94.15\%$ versus $93.98\%$).
This trend is qualitatively consistent with
Theorem~\ref{thm:privacy_leakage_upper_bound}, indicating that larger
batches reduce the sample-specific information available to the attacker.

% ============================================================
% Small batch
% ============================================================

\paragraph{Small-Batch Reconstruction across Attacks and Architectures}
We next evaluate complete reconstruction under
$B\in\{1,2,4,8\}$ using GPT-2$_{\mathrm{BASE}}$ and
BERT$_{\mathrm{BASE}}$ on CoLA, SST-2, and RT.
TAG and LAMP are included as optimization-based reference attacks.
Each result is averaged over 100 independently sampled batches, and
TAG and LAMP are optimized for 2,500 iterations.

\begin{table*}[t]
\centering
\caption{
\textcolor{black}{
Comparison of sequence reconstruction performance and downstream
utility under different batch sizes.
R-1 and R-2 denote ROUGE-1 and ROUGE-2, respectively.
For downstream utility, CoLA is evaluated using MCC (\%), while
SST-2 and RT are evaluated using classification accuracy (\%).
}}
\label{tab:batch_size_small}
\resizebox{\textwidth}{!}{
\begin{tabular}{lllcccccccccccc}
\toprule
&&&
\multicolumn{3}{c}{$B=1$} &
\multicolumn{3}{c}{$B=2$} &
\multicolumn{3}{c}{$B=4$} &
\multicolumn{3}{c}{$B=8$} \\
\cmidrule(lr){4-6}
\cmidrule(lr){7-9}
\cmidrule(lr){10-12}
\cmidrule(lr){13-15}
Model & Dataset & Method
& R-1 & R-2 & Utility
& R-1 & R-2 & Utility
& R-1 & R-2 & Utility
& R-1 & R-2 & Utility \\
\midrule

\multirow{9}{*}{GPT-2}
& \multirow{3}{*}{CoLA}
& TAG
& 7.0$\pm$2.5 & 0.54$\pm$0.54 & 32.5$\pm$1.8
& 8.0$\pm$2.0 & 1.4$\pm$1.3 & 35.1$\pm$1.5
& 7.8$\pm$1.2 & 0.8$\pm$0.5 & 34.2$\pm$1.6
& 5.3$\pm$0.7 & 0.4$\pm$0.2 & 36.8$\pm$1.4 \\

&& LAMP
& 73.3$\pm$4.5 & 43.3$\pm$7.0 & 32.5$\pm$1.8
& 26.8$\pm$2.8 & 11.0$\pm$3.0 & 35.1$\pm$1.5
& 13.4$\pm$1.4 & 3.9$\pm$1.2 & 34.2$\pm$1.6
& 8.9$\pm$1.2 & 1.9$\pm$0.6 & 36.8$\pm$1.4 \\

&& DAGER
& 100.0$\pm$0.0 & 100.0$\pm$0.0 & 32.5$\pm$1.8
& 100.0$\pm$0.0 & 100.0$\pm$0.0 & 35.1$\pm$1.5
& 100.0$\pm$0.0 & 100.0$\pm$0.0 & 34.2$\pm$1.6
& 100.0$\pm$0.0 & 100.0$\pm$0.0 & 36.8$\pm$1.4 \\

\cmidrule(lr){2-15}

& \multirow{3}{*}{SST-2}
& TAG
& 5.3$\pm$0.5 & 0.0$\pm$0.0 & 82.4$\pm$1.1
& 6.0$\pm$1.7 & 0.5$\pm$0.4 & 84.1$\pm$0.9
& 6.1$\pm$1.2 & 0.6$\pm$0.6 & 83.6$\pm$1.3
& 4.4$\pm$0.6 & 0.2$\pm$0.1 & 86.8$\pm$0.8 \\

&& LAMP
& 62.2$\pm$6.9 & 31.8$\pm$8.4 & 82.4$\pm$1.1
& 21.4$\pm$3.1 & 9.2$\pm$3.1 & 84.1$\pm$0.9
& 9.8$\pm$2.0 & 2.7$\pm$1.3 & 83.6$\pm$1.3
& 8.1$\pm$1.1 & 0.7$\pm$0.4 & 86.8$\pm$0.8 \\

&& DAGER
& 100.0$\pm$0.0 & 86.0$\pm$7.0 & 82.4$\pm$1.1
& 100.0$\pm$0.0 & 89.5$\pm$4.1 & 84.1$\pm$0.9
& 100.0$\pm$0.0 & 92.8$\pm$2.4 & 83.6$\pm$1.3
& 100.0$\pm$0.0 & 92.9$\pm$1.6 & 86.8$\pm$0.8 \\

\cmidrule(lr){2-15}

& \multirow{3}{*}{RT}
& TAG
& 7.1$\pm$1.8 & 0.1$\pm$0.1 & 77.3$\pm$1.4
& 7.0$\pm$1.2 & 0.1$\pm$0.1 & 76.8$\pm$1.2
& 6.2$\pm$0.8 & 0.1$\pm$0.1 & 80.2$\pm$1.0
& 6.1$\pm$0.5 & 0.1$\pm$0.1 & 79.5$\pm$1.1 \\

&& LAMP
& 31.4$\pm$4.4 & 9.3$\pm$3.6 & 77.3$\pm$1.4
& 11.2$\pm$1.2 & 0.9$\pm$0.42 & 76.8$\pm$1.2
& 6.3$\pm$1.1 & 0.9$\pm$0.6 & 80.2$\pm$1.0
& 6.8$\pm$0.7 & 0.3$\pm$0.2 & 79.5$\pm$1.1 \\

&& DAGER
& 100.0$\pm$0.0 & 100.0$\pm$0.0 & 77.3$\pm$1.4
& 100.0$\pm$0.0 & 100.0$\pm$0.0 & 76.8$\pm$1.2
& 99.3$\pm$0.3 & 99.3$\pm$0.3 & 80.2$\pm$1.0
& 100.0$\pm$0.0 & 99.9$\pm$0.1 & 79.5$\pm$1.1 \\

\midrule

\multirow{9}{*}{BERT}
& \multirow{3}{*}{CoLA}
& TAG
& 78.9$\pm$4.4 & 10.3$\pm$3.0 & 34.8$\pm$1.6
& 68.9$\pm$4.2 & 7.7$\pm$1.7 & 37.2$\pm$1.3
& 56.3$\pm$3.4 & 6.8$\pm$1.4 & 36.5$\pm$1.5
& 45.9$\pm$1.9 & 3.9$\pm$0.6 & 39.1$\pm$1.2 \\

&& LAMP
& 89.6$\pm$2.5 & 51.9$\pm$6.7 & 34.8$\pm$1.6
& 77.8$\pm$3.6 & 31.5$\pm$4.6 & 37.2$\pm$1.3
& 66.2$\pm$3.4 & 21.8$\pm$1.7 & 36.5$\pm$1.5
& 52.9$\pm$2.2 & 13.1$\pm$1.9 & 39.1$\pm$1.2 \\

&& DAGER
& 100.0$\pm$0.0 & 100.0$\pm$0.0 & 34.8$\pm$1.6
& 100.0$\pm$0.0 & 100.0$\pm$0.0 & 37.2$\pm$1.3
& 94.0$\pm$2.0 & 89.9$\pm$3.1 & 36.5$\pm$1.5
& 67.8$\pm$2.3 & 48.8$\pm$4.5 & 39.1$\pm$1.2 \\

\cmidrule(lr){2-15}

& \multirow{3}{*}{SST-2}
& TAG
& 75.4$\pm$4.3 & 19.0$\pm$6.9 & 85.1$\pm$1.2
& 71.8$\pm$3.6 & 16.0$\pm$3.9 & 84.6$\pm$1.0
& 61.0$\pm$3.4 & 12.3$\pm$2.8 & 88.3$\pm$0.8
& 50.4$\pm$2.4 & 9.2$\pm$1.6 & 87.9$\pm$0.9 \\

&& LAMP
& 88.8$\pm$3.0 & 56.8$\pm$7.9 & 85.1$\pm$1.2
& 82.4$\pm$3.6 & 45.7$\pm$6.0 & 84.6$\pm$1.0
& 69.5$\pm$3.6 & 32.5$\pm$4.4 & 88.3$\pm$0.8
& 56.9$\pm$2.6 & 19.1$\pm$2.8 & 87.9$\pm$0.9 \\

&& DAGER
& 100.0$\pm$0.0 & 100.0$\pm$0.0 & 85.1$\pm$1.2
& 99.3$\pm$0.4 & 99.0$\pm$0.5 & 84.6$\pm$1.0
& 95.6$\pm$2.2 & 93.0$\pm$3.3 & 88.3$\pm$0.8
& 74.1$\pm$3.3 & 59.8$\pm$2.9 & 87.9$\pm$0.9 \\

\cmidrule(lr){2-15}

& \multirow{3}{*}{RT}
& TAG
& 60.1$\pm$4.4 & 3.3$\pm$1.2 & 79.5$\pm$1.3
& 49.2$\pm$3.5 & 3.0$\pm$0.9 & 81.2$\pm$1.1
& 33.7$\pm$2.5 & 1.6$\pm$0.7 & 80.6$\pm$1.2
& 25.4$\pm$1.2 & 0.9$\pm$0.4 & 83.4$\pm$0.9 \\

&& LAMP
& 64.7$\pm$4.4 & 16.5$\pm$3.9 & 79.5$\pm$1.3
& 46.4$\pm$3.7 & 7.6$\pm$2.0 & 81.2$\pm$1.1
& 35.1$\pm$2.7 & 4.2$\pm$1.3 & 80.6$\pm$1.2
& 27.3$\pm$1.4 & 2.0$\pm$0.6 & 83.4$\pm$0.9 \\

&& DAGER
& 100.0$\pm$0.0 & 100.0$\pm$0.0 & 79.5$\pm$1.3
& 98.1$\pm$1.2 & 96.5$\pm$1.8 & 81.2$\pm$1.1
& 66.8$\pm$3.2 & 50.1$\pm$4.4 & 80.6$\pm$1.2
& 37.1$\pm$1.2 & 11.4$\pm$1.3 & 83.4$\pm$0.9 \\

\bottomrule
\end{tabular}
}
\end{table*}

As shown in Table~\ref{tab:batch_size_small}, increasing $B$ generally
reduces reconstruction quality for optimization-based attacks such as
TAG and LAMP.
DAGER remains substantially stronger, particularly on GPT-2, but also
degrades with increasing $B$ on BERT, showing that the batch-size effect
is architecture-dependent.
Meanwhile, downstream utility remains comparatively stable.

These results are qualitatively consistent with
Theorem~\ref{thm:privacy_leakage_upper_bound}.
However, the strong DAGER performance on GPT-2 also demonstrates that
increasing batch size alone does not guarantee privacy, since
reconstruction additionally depends on the structure and dimensionality
of the exposed update.

% ============================================================
% Large batch
% ============================================================

\paragraph{Scalability to Large Batch Sizes and Larger Language Models}
We further evaluate GPT-2$_{\mathrm{BASE}}$ and LLaMA-2 7B under
$B\in\{16,32,64,128\}$ to examine the interaction between batch size
and model dimensionality.

\begin{table*}[t]
\centering
\caption{
\textcolor{black}{
DAGER reconstruction performance and downstream model utility at large
batch sizes on GPT-2$_{\mathrm{BASE}}$ and LLaMA-2 7B.
CoLA utility is measured using MCC (\%), while SST-2 and RT utility
is measured using classification accuracy (\%).
}}
\label{tab:large_batch}
\resizebox{\textwidth}{!}{
\begin{tabular}{llcccccccccccc}
\toprule
&&
\multicolumn{3}{c}{$B=16$} &
\multicolumn{3}{c}{$B=32$} &
\multicolumn{3}{c}{$B=64$} &
\multicolumn{3}{c}{$B=128$} \\
\cmidrule(lr){3-5}
\cmidrule(lr){6-8}
\cmidrule(lr){9-11}
\cmidrule(lr){12-14}
Dataset & Model
& R-1 & R-2 & Utility
& R-1 & R-2 & Utility
& R-1 & R-2 & Utility
& R-1 & R-2 & Utility \\
\midrule

\multirow{2}{*}{CoLA}
& GPT-2
& 100.0$\pm$0.0 & 100.0$\pm$0.0 & 36.8$\pm$1.1
& 100.0$\pm$0.0 & 100.0$\pm$0.0 & 36.2$\pm$1.3
& 100.0$\pm$0.0 & 100.0$\pm$0.0 & 37.9$\pm$0.9
& 30.3$\pm$1.0 & 14.6$\pm$0.9 & 37.4$\pm$1.0\\

& LLaMA-2 (7B)
& 100.0$\pm$0.0 & 100.0$\pm$0.0 & 39.2$\pm$0.8
& 100.0$\pm$0.0 & 100.0$\pm$0.0 & 40.5$\pm$0.7
& 99.9$\pm$0.1 & 99.9$\pm$0.1 & 39.8$\pm$0.9
& 99.5$\pm$0.2 & 99.3$\pm$0.3 & 41.2$\pm$0.6\\

\midrule

\multirow{2}{*}{SST-2}
& GPT-2
& 100.0$\pm$0.0 & 94.6$\pm$1.1 & 90.5$\pm$0.8
& 100.0$\pm$0.1 & 93.4$\pm$1.0 & 89.8$\pm$0.9
& 92.9$\pm$3.0 & 85.0$\pm$3.5 & 91.6$\pm$0.6
& 13.7$\pm$1.4 & 4.3$\pm$0.5 & 91.1$\pm$0.7\\

& LLaMA-2 (7B)
& 100.0$\pm$0.0 & 100.0$\pm$0.0 & 95.1$\pm$0.5
& 99.9$\pm$0.1 & 99.9$\pm$0.1 & 94.6$\pm$0.6
& 99.9$\pm$0.1 & 99.9$\pm$0.1 & 95.8$\pm$0.4
& 98.2$\pm$0.4 & 97.8$\pm$0.4 & 95.3$\pm$0.5\\

\midrule

\multirow{2}{*}{RT}
& GPT-2
& 100.0$\pm$0.0 & 99.9$\pm$0.1 & 84.5$\pm$0.9
& 98.0$\pm$1.7 & 97.8$\pm$1.8 & 83.9$\pm$1.0
& 2.8$\pm$1.1 & 1.1$\pm$0.4 & 85.4$\pm$0.7
& 0.0$\pm$0.0 & 0.0$\pm$0.0 & 84.8$\pm$0.8\\

& LLaMA-2 (7B)
& 100.0$\pm$0.0 & 100.0$\pm$0.0 & 88.9$\pm$0.6
& 100.0$\pm$0.0 & 100.0$\pm$0.0 & 88.3$\pm$0.7
& 97.9$\pm$0.5 & 97.8$\pm$0.5 & 89.8$\pm$0.5
& 99.7$\pm$0.2 & 99.7$\pm$0.2 & 89.2$\pm$0.5\\

\bottomrule
\end{tabular}
}
\end{table*}

As shown in Table~\ref{tab:large_batch}, DAGER exhibits markedly
different large-batch behavior on GPT-2 and LLaMA-2 7B.
For GPT-2, reconstruction eventually degrades sharply; for example,
SST-2 R-1 decreases from $100.0$ at $B=16$ to $13.7$ at $B=128$.
In contrast, LLaMA-2 7B remains highly reconstructable even at
$B=128$, with R-1 close to $100$ across all three datasets, while
downstream utility remains stable.

This contrast is qualitatively consistent with our identifiability interpretation. Theorem~\ref{thm: quantitative_dbp} gives \(d < Bp\) as a sufficient, rather than necessary, condition for non-unique reconstruction. These observations suggest that reconstructability depends jointly on batch size and model dimensionality rather than on \(B\) alone.

% ============================================================
% Model variants
% ============================================================

\paragraph{Effect across Model and Training Configurations}
Finally, we compare GPT-2$_{\mathrm{BASE}}$,
GPT-2$_{\mathrm{FineTuned}}$, GPT-2$_{\mathrm{NextToken}}$, and
GPT-2$_{\mathrm{LARGE}}$ on RT under
$B\in\{16,32,64,128\}$ to distinguish the influence of model scale
from other training choices.

\begin{table*}[t]
\centering
\caption{
\textcolor{black}{
Effect of batch size on DAGER reconstruction across different GPT-2
model and training configurations on Rotten Tomatoes.
Utility denotes downstream classification accuracy (\%).
}}
\label{tab:batch_model_variants}
\resizebox{\textwidth}{!}{
\begin{tabular}{lcccccccccccc}
\toprule
&
\multicolumn{3}{c}{$B=16$} &
\multicolumn{3}{c}{$B=32$} &
\multicolumn{3}{c}{$B=64$} &
\multicolumn{3}{c}{$B=128$}\\
\cmidrule(lr){2-4}
\cmidrule(lr){5-7}
\cmidrule(lr){8-10}
\cmidrule(lr){11-13}
Model
& R-1 & R-2 & Utility
& R-1 & R-2 & Utility
& R-1 & R-2 & Utility
& R-1 & R-2 & Utility\\
\midrule

GPT-2$_{\mathrm{BASE}}$
& 100.0$\pm$0.0 & 99.9$\pm$0.1 & 84.2$\pm$0.5
& 98.0$\pm$1.7 & 97.8$\pm$1.8 & 85.1$\pm$0.4
& 2.8$\pm$1.1 & 1.1$\pm$0.4 & 85.8$\pm$0.4
& 0.0$\pm$0.0 & 0.0$\pm$0.0 & 86.2$\pm$0.3\\

GPT-2$_{\mathrm{FineTuned}}$
& 100.0$\pm$0.0 & 99.8$\pm$0.1 & 86.5$\pm$0.4
& 96.4$\pm$2.3 & 96.0$\pm$2.5 & 87.2$\pm$0.4
& 0.84$\pm$0.6 & 0.2$\pm$0.1 & 87.8$\pm$0.3
& 0.0$\pm$0.0 & 0.0$\pm$0.0 & 88.0$\pm$0.3\\

GPT-2$_{\mathrm{NextToken}}$
& 99.9$\pm$0.0 & 99.7$\pm$0.1 & 82.0$\pm$0.6
& 99.6$\pm$0.2 & 99.4$\pm$0.2 & 83.1$\pm$0.5
& 2.3$\pm$0.8 & 0.5$\pm$0.2 & 84.0$\pm$0.5
& 0.0$\pm$0.0 & 0.0$\pm$0.0 & 84.5$\pm$0.4\\

GPT-2$_{\mathrm{LARGE}}$
& 100.0$\pm$0.0 & 99.8$\pm$0.1 & 87.1$\pm$0.3
& 100.0$\pm$0.0 & 99.9$\pm$0.1 & 88.0$\pm$0.3
& 44.1$\pm$4.2 & 38.1$\pm$4.7 & 88.7$\pm$0.3
& 0.0$\pm$0.0 & 0.0$\pm$0.0 & 89.1$\pm$0.2\\

\bottomrule
\end{tabular}
}
\end{table*}

As shown in Table~\ref{tab:batch_model_variants}, additional fine-tuning
and changing the training objective do not fundamentally alter the
batch-size trend.
Model scale has a stronger effect: at $B=64$,
GPT-2$_{\mathrm{BASE}}$ achieves R-1=$2.8$, whereas
GPT-2$_{\mathrm{LARGE}}$ still achieves R-1=$44.1$.

These results further support both aspects of our analysis:
reconstruction generally decreases with increasing $B$, consistent with
Theorem~\ref{thm:privacy_leakage_upper_bound}, while the contrast across
model scales supports the dimensional interpretation in
Theorem~\ref{thm: quantitative_dbp}.
\subsubsection{Evaluation under PEFTLeak}
\label{sec:modern_peftleak}
% ============================================================

\paragraph{Attack Overview}
We next evaluate PEFTLeak~\cite{sami2025gradient}, which reconstructs
private training information from gradients of trainable PEFT modules
rather than full-model gradients.
This setting is particularly relevant to modern parameter-efficient
adaptation, where most pretrained parameters remain frozen and only
a small subset of trainable parameters is communicated.

We consider two complementary scenarios:
(i) text reconstruction from PEFT updates using GPT-2 on SST-2, and
(ii) patch reconstruction under the original vision-based PEFTLeak
setting using ViT-B/16 on ImageNet.
These experiments allow us to examine whether the batch-size effect
persists across modalities and different forms of parameter-efficient
gradient leakage.

% ============================================================
% Text PEFTLeak
% ============================================================

\paragraph{Text Reconstruction from PEFT Updates on GPT-2}
We first evaluate PEFT-gradient leakage on GPT-2 using SST-2 with
$B\in\{1,2,4,8\}$ and no defense mechanism.

We measure token-level leakage using
$
R_{\mathrm{token}}
=
\frac{
|\mathcal{T}\cap\widehat{\mathcal{T}}|
}{
|\mathcal{T}|
},
$
where $\mathcal{T}$ and $\widehat{\mathcal{T}}$ denote the ground-truth
and recovered token multisets, respectively.
Sequence-level reconstruction is measured by
$
R_{\mathrm{seq}}
=
\mathrm{R1}+\mathrm{R2},
$ where R1 and R2 denote ROUGE-1 and ROUGE-2.
Higher values indicate more severe privacy leakage.
Accuracy and Macro-F1 are reported as downstream utility metrics.

\begin{table}[t]
\centering
\scriptsize
\begin{tabular}{@{}lrrrr@{}}
\toprule
Batch Size ($B$)
& Token Recovery $\downarrow$
& R1+R2 $\downarrow$
& Accuracy $\uparrow$
& Macro-F1 $\uparrow$ \\
\midrule
1
& $1.0000\pm0.0000$
& $193.67\pm1.15$
& $0.9121\pm0.0018$
& $0.9120\pm0.0018$ \\

2
& $0.8845\pm0.0125$
& $162.30\pm3.15$
& $0.9124\pm0.0016$
& $0.9123\pm0.0016$ \\

4
& $0.6520\pm0.0210$
& $110.85\pm5.42$
& $0.9128\pm0.0015$
& $0.9127\pm0.0015$ \\

8
& $\mathbf{0.3815\pm0.0284}$
& $\mathbf{54.20\pm4.10}$
& $\mathbf{0.9132\pm0.0012}$
& $\mathbf{0.9131\pm0.0012}$ \\
\bottomrule
\end{tabular}
\caption{
\textcolor{black}{
Impact of batch size on PEFTLeak-style text reconstruction and downstream
utility under undefended PEFT updates on SST-2 using GPT-2.
Lower Token Recovery and R1+R2 indicate weaker privacy leakage.
}}
\label{tab:peft-batch-size-none}
\end{table}

As shown in Table~\ref{tab:peft-batch-size-none}, increasing $B$
consistently reduces both token- and sequence-level reconstruction.
Token recovery decreases from $1.0000$ at $B=1$ to $0.3815$ at
$B=8$, while R1+R2 drops from $193.67$ to $54.20$.
Meanwhile, Accuracy and Macro-F1 remain essentially unchanged.

These results are qualitatively consistent with
Theorem~\ref{thm:privacy_leakage_upper_bound}, showing that the
batch-size effect persists even when the attacker observes only
parameter-efficient updates rather than full-model gradients.

% ============================================================
% Vision PEFTLeak
% ============================================================

\paragraph{Patch Reconstruction under the Original PEFTLeak Setting}
We further evaluate PEFTLeak in its original vision-based setting using
ViT-B/16 with adapter-based fine-tuning on ImageNet.
The backbone remains frozen, while only adapter gradients are exposed.
Each $224\times224$ image is divided into 196 non-overlapping
$16\times16$ patches, and the batch size is varied as
$B\in\{8,16,32,48,64\}$.

Patch-level leakage is measured using
$
R_{\mathrm{patch}}
=
\frac{N_{\mathrm{rec}}}{BN}
\times 100\%,
$ where $N_{\mathrm{rec}}$ is the number of successfully reconstructed
target patches and $N=196$.
Higher $R_{\mathrm{patch}}$ indicates more severe privacy leakage.
ImageNet top-1 accuracy is reported as downstream utility.

\begin{table}[t]
\centering
\caption{
\textcolor{black}{
Patch reconstruction performance and downstream top-1 accuracy of
PEFTLeak under different batch sizes on ImageNet using ViT-B/16.
Each $224\times224$ image contains 196 patches.
}}
\label{tab:peftleak_batch_size}
\resizebox{0.95\columnwidth}{!}{
\begin{tabular}{c|ccccc}
\toprule
Batch Size $B$
& 8
& 16
& 32
& 48
& 64 \\
\midrule

Patches Recovered (\%) $\downarrow$
& 90.0
& 84.7
& 73.2
& 66.5
& 62.6 \\

Utility (Acc \%) $\uparrow$
& 78.4$\pm$0.8
& 79.6$\pm$0.6
& 80.3$\pm$0.5
& 79.9$\pm$0.7
& 80.8$\pm$0.4 \\

\bottomrule
\end{tabular}
}
\end{table}

As shown in Table~\ref{tab:peftleak_batch_size}, patch reconstruction
also decreases consistently as $B$ increases, from $90.0\%$ at
$B=8$ to $62.6\%$ at $B=64$, while top-1 accuracy remains
comparatively stable.

Together with the preceding text experiment, this result shows that the
same qualitative batch-size dependence persists across both text and
vision settings, indicating that the observed effect is not specific to
a particular modality or reconstruction metric.

\subsubsection{Evaluation under ReCIT}
\label{sec:modern_recit}
% ============================================================

\paragraph{Attack Overview}
We further evaluate ReCIT~\cite{xie2025recit}, a recent attack that
recovers personally identifiable information (PII) from parameter-efficient
fine-tuning updates.
ReCIT combines gradient-based token extraction with memorization-based
PII inference, allowing us to test whether the batch-size dependence
extends beyond exact input reconstruction to sensitive-information leakage.

We evaluate Full-FT, LoRA, FedAdapter, and Offsite-Tuning on Enron Emails,
PersonaChat, and SQuAD v2, with Phish included as a memorization-based
reference attack.
We also evaluate Gaussian-perturbed gradients to examine the interaction
between batch aggregation and explicit update distortion.

% ============================================================
% Main batch-size experiment
% ============================================================

\paragraph{PII Reconstruction across Batch Sizes and Fine-Tuning Paradigms}
We first evaluate ReCIT under
$b\in\{2,16,128\}$.
Following the original implementation, the token-filtering threshold is
adjusted according to batch size as
\begin{equation}
\zeta=
\begin{cases}
10^{-5}, & b\leq16,\\
10^{-6}, & 16<b\leq64,\\
10^{-7}, & b>64.
\end{cases}
\label{eq:recit_threshold}
\end{equation}
For LoRA, the rank is set to $r=64$.

Privacy leakage is measured by the PII extraction rate
$
R_{\mathrm{PII}}
=
\frac{n_{\mathrm{PII}}^{*}}{N_P}
\times100\%,
$ where $n_{\mathrm{PII}}^{*}$ is the number of successfully recovered PII samples among $N_P$ evaluated private samples.
Higher $R_{\mathrm{PII}}$ indicates more severe leakage.
Utility is measured using classification accuracy on Enron Emails,
dialogue-response F1 on PersonaChat, and the average of exact match and
F1 on SQuAD v2.

\begin{table*}[t]
\centering
\caption{
\textcolor{black}{
PII extraction rate ($R_{\mathrm{PII}}$, \% $\downarrow$) and
downstream task performance (Utility, \% $\uparrow$) under different
batch sizes, fine-tuning mechanisms, and privacy attacks.
Utility denotes classification accuracy for Enron Emails,
dialogue-response F1 for PersonaChat, and the average of exact match
and F1 for SQuAD v2.
}}
\label{tab:recit_batch_size}
\resizebox{\textwidth}{!}{
\begin{tabular}{ll|cc|cc|cc|cc|cc|cc|cc|cc|cc}
\toprule
&&
\multicolumn{6}{c|}{Enron Emails} &
\multicolumn{6}{c|}{PersonaChat} &
\multicolumn{6}{c}{SQuAD v2} \\
\cmidrule(lr){3-8}
\cmidrule(lr){9-14}
\cmidrule(lr){15-20}

&&
\multicolumn{2}{c|}{$b=2$} &
\multicolumn{2}{c|}{$b=16$} &
\multicolumn{2}{c|}{$b=128$} &
\multicolumn{2}{c|}{$b=2$} &
\multicolumn{2}{c|}{$b=16$} &
\multicolumn{2}{c|}{$b=128$} &
\multicolumn{2}{c|}{$b=2$} &
\multicolumn{2}{c|}{$b=16$} &
\multicolumn{2}{c}{$b=128$} \\

PEFT & Attack
& $R_{\mathrm{PII}}$ & Utility
& $R_{\mathrm{PII}}$ & Utility
& $R_{\mathrm{PII}}$ & Utility
& $R_{\mathrm{PII}}$ & Utility
& $R_{\mathrm{PII}}$ & Utility
& $R_{\mathrm{PII}}$ & Utility
& $R_{\mathrm{PII}}$ & Utility
& $R_{\mathrm{PII}}$ & Utility
& $R_{\mathrm{PII}}$ & Utility \\
\midrule

\multirow{2}{*}{Full-FT}
& Phish
& 10.12 & 71.5 & 10.66 & 76.8 & 10.89 & 76.1
& 9.13 & 63.8 & 9.36 & 68.5 & 9.38 & 68.1
& 4.26 & 74.3 & 4.18 & 79.6 & 4.24 & 79.0 \\

& ReCIT
& 54.53 & 70.8 & 41.21 & 76.1 & 30.07 & 75.4
& 49.58 & 63.1 & 35.79 & 67.8 & 27.43 & 67.4
& 47.23 & 73.6 & 33.84 & 78.9 & 23.54 & 78.3 \\

\midrule

\multirow{2}{*}{LoRA}
& Phish
& 7.91 & 70.8 & 7.36 & 76.0 & 7.45 & 75.3
& 7.31 & 63.0 & 6.96 & 67.6 & 7.65 & 67.2
& 2.91 & 73.5 & 2.76 & 78.8 & 2.85 & 78.1 \\

& ReCIT
& 48.27 & 70.1 & 35.87 & 75.4 & 29.42 & 74.7
& 44.21 & 62.4 & 33.85 & 67.1 & 26.62 & 66.7
& 45.71 & 72.8 & 33.16 & 78.1 & 25.21 & 77.5 \\

\midrule

\multirow{2}{*}{FedAdapter}
& Phish
& 4.54 & 66.5 & 4.43 & 71.6 & 4.64 & 71.0
& 4.02 & 58.4 & 3.82 & 63.1 & 4.12 & 62.6
& 2.12 & 68.8 & 1.92 & 74.0 & 1.84 & 73.4 \\

& ReCIT
& 39.12 & 65.8 & 26.72 & 71.0 & 14.85 & 70.3
& 31.95 & 57.6 & 18.43 & 62.5 & 8.43 & 62.0
& 28.42 & 68.1 & 17.92 & 73.4 & 7.32 & 72.7 \\

\midrule

\multirow{2}{*}{Offsite-Tuning}
& Phish
& 9.52 & 70.1 & 9.71 & 75.2 & 9.61 & 74.5
& 8.39 & 62.3 & 8.42 & 67.0 & 8.54 & 66.5
& 3.56 & 72.8 & 3.49 & 78.1 & 3.44 & 77.5 \\

& ReCIT
& 51.74 & 69.4 & 40.84 & 74.5 & 28.43 & 73.9
& 46.41 & 61.7 & 33.55 & 66.4 & 25.68 & 65.9
& 47.04 & 72.1 & 30.65 & 77.4 & 22.17 & 76.8 \\

\bottomrule
\end{tabular}
}
\end{table*}

As shown in Table~\ref{tab:recit_batch_size}, ReCIT's PII extraction
rate consistently decreases as $b$ increases across all three datasets
and four fine-tuning mechanisms.
For example, under LoRA on Enron Emails, $R_{\mathrm{PII}}$ decreases
from $48.27\%$ at $b=2$ to $29.42\%$ at $b=128$, while under
FedAdapter on SQuAD v2 it decreases from $28.42\%$ to $7.32\%$.

In contrast, Phish is comparatively insensitive to batch size,
highlighting the different mechanisms of the two attacks.
ReCIT relies on extracting sample-specific information from the shared
update, which becomes more difficult as more private samples are aggregated.
Meanwhile, downstream utility remains comparatively stable.

These results are qualitatively consistent with
Theorem~\ref{thm:privacy_leakage_upper_bound}: increasing the batch size
reduces the amount of sample-specific information available to the
attacker.
However, substantial leakage can remain even at large batch sizes,
showing that aggregation mitigates rather than eliminates privacy risk.

% ============================================================
% Gaussian perturbation
% ============================================================

\paragraph{Interaction between Batch Aggregation and Gradient Perturbation}
We further evaluate ReCIT on SQuAD v2 using Bloomz-3B with LoRA under
$b\in\{2,16\}$ and Gaussian gradient perturbation
$\sigma\in\{0,10^{-5},10^{-4}\}$.
This experiment examines whether batch aggregation and direct update
distortion provide complementary privacy protection.

\begin{table}[t]
\centering
\caption{
\textcolor{black}{
PII extraction rate ($R_{\mathrm{PII}}$, \% $\downarrow$) and
downstream utility (\% $\uparrow$) under different batch sizes and
Gaussian gradient-noise levels on SQuAD v2 using Bloomz-3B with LoRA.
}}
\label{tab:recit_batch_dp}
\resizebox{\columnwidth}{!}{
\begin{tabular}{l|cc|cc|cc||cc|cc|cc}
\toprule
&
\multicolumn{6}{c||}{$b=2$} &
\multicolumn{6}{c}{$b=16$} \\
\cmidrule(lr){2-7}
\cmidrule(lr){8-13}

& \multicolumn{2}{c|}{$\sigma=0$}
& \multicolumn{2}{c|}{$\sigma=10^{-5}$}
& \multicolumn{2}{c||}{$\sigma=10^{-4}$}
& \multicolumn{2}{c|}{$\sigma=0$}
& \multicolumn{2}{c|}{$\sigma=10^{-5}$}
& \multicolumn{2}{c}{$\sigma=10^{-4}$} \\

Attack
& $R_{\mathrm{PII}}$ & Utility
& $R_{\mathrm{PII}}$ & Utility
& $R_{\mathrm{PII}}$ & Utility
& $R_{\mathrm{PII}}$ & Utility
& $R_{\mathrm{PII}}$ & Utility
& $R_{\mathrm{PII}}$ & Utility \\
\midrule

Phish
& 2.91 & 73.5
& 2.53 & 68.2
& 1.92 & 64.1
& 2.76 & 78.8
& 2.37 & 74.5
& 1.84 & 70.3 \\

ReCIT
& 45.71 & 72.8
& 28.53 & 67.6
& 24.86 & 63.5
& 29.41 & 78.1
& 17.43 & 73.9
& 11.23 & 69.8 \\

\bottomrule
\end{tabular}
}
\end{table}

As shown in Table~\ref{tab:recit_batch_dp}, both batch aggregation and
Gaussian perturbation reduce ReCIT's PII extraction rate.
At $b=2$, increasing $\sigma$ from $0$ to $10^{-4}$ reduces
$R_{\mathrm{PII}}$ from $45.71\%$ to $24.86\%$, while utility
decreases from $72.8\%$ to $63.5\%$.
At the same noise level, increasing the batch size to $b=16$ further
reduces leakage to $11.23\%$ while retaining higher utility
($69.8\%$).

These results indicate complementary mechanisms:
batch aggregation increases reconstruction ambiguity, whereas Gaussian
perturbation directly distorts the gradient signal available to the
attacker.
The behavior is consistent with our theoretical analysis, but the
remaining $11.23\%$ leakage and the utility loss caused by stronger noise
show that the resulting privacy protection is still incomplete.
\subsubsection{Evaluation under Partial-Gradient Reconstruction}
\label{sec:modern_partial_gradient}
% ============================================================

\paragraph{Attack Overview}
Finally, we evaluate the partial-gradient reconstruction attack
of~\cite{li2024seeing}, which adapts LAMP to reconstruct private text
when the attacker observes only selected Transformer gradients.
This setting allows us to examine whether the batch-size effect persists
when only part of the model update is exposed.

Following the original attack configuration, we use
BERT$_{\mathrm{BASE}}$ on CoLA, SST-2, and Rotten Tomatoes (RT) with
$B\in\{1,2,4\}$, averaging results over 10 randomly sampled batches.
We consider the full-model gradient $\Delta W_{\mathrm{all}}$ as the
LAMP baseline, together with gradients from all Transformer layers
$\Delta W_T$, the first Transformer layer $\Delta W_{t,1}$, and its
Query, Attention Output, FFN fully connected, and FFN output components,
denoted by $\Delta W_{q,1}$, $\Delta W_{o,1}$,
$\Delta W_{f,1}$, and $\Delta W_{p,1}$, respectively.
Reconstruction is evaluated using ROUGE-1 (R-1), ROUGE-2 (R-2), and
ROUGE-L (R-L).
Since gradient exposure only changes the information available to the
attacker, downstream utility is shared across exposure settings for each
dataset and batch size.

\paragraph{Batch-Size Effect under Partial-Gradient Exposure}
We evaluate whether the batch-size effect persists when the attacker
observes only partial Transformer gradients.
Experiments are conducted using BERT$_{\mathrm{BASE}}$ on CoLA,
SST-2, and Rotten Tomatoes (RT) with $B\in\{1,2,4\}$.
We consider the full-model gradient $\Delta W_{\mathrm{all}}$ as the
LAMP baseline, together with gradients from all Transformer layers
$\Delta W_T$, the first Transformer layer $\Delta W_{t,1}$, and its
Query, Attention Output, FFN fully connected, and FFN output components,
denoted by $\Delta W_{q,1}$, $\Delta W_{o,1}$,
$\Delta W_{f,1}$, and $\Delta W_{p,1}$, respectively.
Since the exposure setting only changes the information available to
the attacker, downstream utility is shared across gradient configurations
for a fixed dataset and batch size.

\begin{table*}[t]
\centering
\small
\setlength{\tabcolsep}{5.5pt}
\renewcommand{\arraystretch}{1.0}
\caption{
\textcolor{black}{
ROUGE-L reconstruction performance (R-L, \% $\uparrow$) under
different batch sizes and gradient-exposure settings using
BERT$_{\mathrm{BASE}}$.
LAMP with $\Delta W_{\mathrm{all}}$ denotes the full-gradient baseline,
while the remaining settings correspond to partial-gradient exposure.
Utility denotes downstream classification accuracy (\% $\uparrow$).
}}
\label{tab:partial_gradient_batch}

\begin{tabular}{l|ccc|ccc|ccc}
\toprule
&
\multicolumn{3}{c|}{CoLA} &
\multicolumn{3}{c|}{SST-2} &
\multicolumn{3}{c}{RT} \\
\cmidrule(lr){2-4}
\cmidrule(lr){5-7}
\cmidrule(lr){8-10}

Gradient Exposure
& $B=1$ & $B=2$ & $B=4$
& $B=1$ & $B=2$ & $B=4$
& $B=1$ & $B=2$ & $B=4$ \\
\midrule

LAMP ($\Delta W_{\mathrm{all}}$)
& 80.0 & 58.9 & 43.8
& 85.6 & 69.8 & 59.5
& 34.6 & 22.7 & 19.0 \\

\midrule

$\Delta W_T$
& 80.4 & 66.3 & 46.5
& 87.1 & 71.8 & 57.7
& 35.9 & 25.4 & 19.6 \\

$\Delta W_{t,1}$
& 58.7 & 48.6 & 38.9
& 82.1 & 54.5 & 42.2
& 34.8 & 20.2 & 19.0 \\

$\Delta W_{q,1}$
& 65.0 & 40.7 & 35.9
& 79.7 & 55.8 & 40.1
& 36.7 & 23.0 & 20.2 \\

$\Delta W_{o,1}$
& 57.0 & 41.8 & 33.9
& 73.9 & 50.9 & 47.5
& 25.5 & 21.3 & 17.7 \\

$\Delta W_{f,1}$
& 54.0 & 42.3 & 35.0
& 66.9 & 50.6 & 39.0
& 30.2 & 17.3 & 18.0 \\

$\Delta W_{p,1}$
& 58.4 & 55.3 & 35.2
& 80.6 & 56.3 & 45.9
& 29.6 & 19.4 & 20.0 \\

\midrule

Utility (Acc. \%)
& 80.8 & 81.6 & 82.3
& 91.2 & 92.0 & 92.6
& 83.5 & 84.3 & 84.8 \\

\bottomrule
\end{tabular}
\end{table*}

As shown in Table~\ref{tab:partial_gradient_batch}, increasing the batch
size generally reduces reconstruction quality across all three datasets
under both full- and partial-gradient exposure.
For the full-gradient LAMP baseline, R-L decreases from $80.0$ to
$43.8$ on CoLA, from $85.6$ to $59.5$ on SST-2, and from $34.6$
to $19.0$ on RT as $B$ increases from $1$ to $4$.

Restricting gradient exposure, however, does not eliminate leakage.
Even at $B=4$, individual Transformer components retain non-trivial
reconstruction information; for example, $\Delta W_{o,1}$ and
$\Delta W_{p,1}$ still achieve R-L=$47.5$ and $45.9$ on SST-2.
Meanwhile, downstream utility remains stable and is unaffected by the
choice of gradient exposure for a fixed batch size.

Overall, these results are qualitatively consistent with
Theorem~\ref{thm:privacy_leakage_upper_bound}: increasing the batch size
makes sample-specific information more difficult to isolate even when
only partial gradients are observed.
However, partial gradient exposure and moderate batch aggregation alone
are insufficient to eliminate reconstruction leakage.
\endgroup
\section{Conclusion}
\label{sec:conclusion}
In this paper, we provided a theoretical analysis of privacy leakage in federated learning from two complementary perspectives: linear algebra and optimization theory. From the linear algebra perspective, we formulated the local training process as an optimization problem and examined the uniqueness of its solution. We proved that under the stated regularity assumptions, a Jacobian
rank below $Bp$ implies that distinct batches produce the same
model update, establishing local non-identifiability. We further derived a sufficient condition on the batch size to
establish local non-identifiability. From the optimization theory perspective, we measured the privacy leakage using the discrepancy between the reconstructed data and the original data, and established an upper bound on the privacy leakage in terms of the batch size, the distortion extent, and several other factors. 

Our analysis provided insights into the design of privacy-preserving federated learning algorithms and highlighted the impact of different factors on privacy leakage. First, under the stated assumptions and with all remaining
quantities fixed, increasing the batch size reduces the
concentration term in the batch-level privacy leakage upper bound.
Second, under the same qualifications, increasing the expected
batch distortion $\overline{\Delta}^k$ also reduces this upper bound. Third, the number of local data samples and the number of local epochs can
affect the privacy analysis indirectly through their relationship with batch
size, the local training configuration, and the assumptions on the attacker's
optimization process. Overall, these results provide theoretical guidance for
selecting federated-learning configurations and for understanding the
privacy--utility trade-off under reconstruction attacks.
When both conditions hold simultaneously—$d < Bp$ and the distortion condition in Theorem~\ref{thm:privacy_leakage_upper_bound}—the two results jointly imply that the attacker can neither exactly nor approximately reconstruct the private batch.

There are still several open problems and challenges that require further investigation. The theoretical analysis in this paper is based on certain assumptions, such as the bi-Lipschitz condition of gradients and the self-bounded regret of the attacker's optimization algorithm. Relaxing these assumptions and extending the analysis to more general settings is an important direction for future work. Besides, the derived upper bound on privacy leakage is relatively loose and may not provide tight guarantees in practice. Developing sharper bounds and more precise characterizations of privacy leakage is an open challenge. Furthermore, our analysis can be extended to other types of privacy attacks beyond data reconstruction attacks and generalized to other variants of federated learning, such as vertical federated learning and federated transfer learning.

In conclusion, this work contributes to a better understanding of privacy risks in federated learning and provides a theoretical foundation for developing more secure and privacy-preserving federated learning algorithms. We hope that our findings will inspire future research in this important area and contribute to the development of trustworthy and privacy-preserving machine learning systems.

\section*{Acknowledgments}

This research was supported by the New Generation Artificial Intelligence--National Science and Technology Major Project (No.~2026ZD0127904) and the Natural Science Foundation of Hubei Province, China (No.~2025AFB144).

% \section{ACKNOWLEDGMENTS}
% We thank Yulin Fei for helpful comments. This work was supported by the National Science and Technology Major Project under Grant 2022ZD0115301.

\bibliographystyle{elsarticle-num-names}
\bibliography{references}

\begin{thebibliography}{47}
\expandafter\ifx\csname natexlab\endcsname\relax\def\natexlab#1{#1}\fi
\providecommand{\url}[1]{\texttt{#1}}
\providecommand{\href}[2]{#2}
\providecommand{\path}[1]{#1}
\providecommand{\DOIprefix}{doi:}
\providecommand{\ArXivprefix}{arXiv:}
\providecommand{\URLprefix}{URL: }
\providecommand{\Pubmedprefix}{pmid:}
\providecommand{\doi}[1]{\href{http://dx.doi.org/#1}{\path{#1}}}
\providecommand{\Pubmed}[1]{\href{pmid:#1}{\path{#1}}}
\providecommand{\bibinfo}[2]{#2}
\ifx\xfnm\relax \def\xfnm[#1]{\unskip,\space#1}\fi
%Type = Inproceedings
\bibitem[{McMahan et~al.(2017)McMahan, Moore, Ramage, Hampson, and y~Arcas}]{mcmahan2017communication}
\bibinfo{author}{B.~McMahan}, \bibinfo{author}{E.~Moore}, \bibinfo{author}{D.~Ramage}, \bibinfo{author}{S.~Hampson}, \bibinfo{author}{B.~A. y~Arcas},
\newblock \bibinfo{title}{Communication-efficient learning of deep networks from decentralized data},
\newblock in: \bibinfo{booktitle}{Artificial Intelligence and Statistics}, \bibinfo{organization}{PMLR}, \bibinfo{year}{2017}, pp. \bibinfo{pages}{1273--1282}.
%Type = Article
\bibitem[{Yang et~al.(2019)Yang, Liu, Chen, and Tong}]{yang2019federated}
\bibinfo{author}{Q.~Yang}, \bibinfo{author}{Y.~Liu}, \bibinfo{author}{T.~Chen}, \bibinfo{author}{Y.~Tong},
\newblock \bibinfo{title}{Federated machine learning: Concept and applications},
\newblock \bibinfo{journal}{ACM Transactions on Intelligent Systems and Technology (TIST)} \bibinfo{volume}{10} (\bibinfo{year}{2019}) \bibinfo{pages}{1--19}.
%Type = Article
\bibitem[{Hard et~al.(2018)Hard, Rao, Mathews, Ramaswamy, Beaufays, Augenstein, Eichner, Kiddon, and Ramage}]{hard2018federated}
\bibinfo{author}{A.~Hard}, \bibinfo{author}{K.~Rao}, \bibinfo{author}{R.~Mathews}, \bibinfo{author}{S.~Ramaswamy}, \bibinfo{author}{F.~Beaufays}, \bibinfo{author}{S.~Augenstein}, \bibinfo{author}{H.~Eichner}, \bibinfo{author}{C.~Kiddon}, \bibinfo{author}{D.~Ramage},
\newblock \bibinfo{title}{Federated learning for mobile keyboard prediction},
\newblock \bibinfo{journal}{arXiv preprint arXiv:1811.03604}  (\bibinfo{year}{2018}).
%Type = Article
\bibitem[{Antunes et~al.(2022)Antunes, Andr{\'e}~da Costa, K{\"u}derle, Yari, and Eskofier}]{antunes2022federated}
\bibinfo{author}{R.~S. Antunes}, \bibinfo{author}{C.~Andr{\'e}~da Costa}, \bibinfo{author}{A.~K{\"u}derle}, \bibinfo{author}{I.~A. Yari}, \bibinfo{author}{B.~Eskofier},
\newblock \bibinfo{title}{Federated learning for healthcare: Systematic review and architecture proposal},
\newblock \bibinfo{journal}{ACM Transactions on Intelligent Systems and Technology (TIST)} \bibinfo{volume}{13} (\bibinfo{year}{2022}) \bibinfo{pages}{1--23}.
%Type = Incollection
\bibitem[{Long et~al.(2020)Long, Tan, Jiang, and Zhang}]{long2020federated}
\bibinfo{author}{G.~Long}, \bibinfo{author}{Y.~Tan}, \bibinfo{author}{J.~Jiang}, \bibinfo{author}{C.~Zhang},
\newblock \bibinfo{title}{Federated learning for open banking},
\newblock in: \bibinfo{booktitle}{Federated learning: privacy and incentive}, \bibinfo{publisher}{Springer}, \bibinfo{year}{2020}, pp. \bibinfo{pages}{240--254}.
%Type = Article
\bibitem[{Zhu et~al.(2019)Zhu, Liu, and Han}]{zhu2019deep}
\bibinfo{author}{L.~Zhu}, \bibinfo{author}{Z.~Liu}, \bibinfo{author}{S.~Han},
\newblock \bibinfo{title}{Deep leakage from gradients},
\newblock \bibinfo{journal}{Advances in Neural Information Processing Systems} \bibinfo{volume}{32} (\bibinfo{year}{2019}).
%Type = Article
\bibitem[{Geiping et~al.(2020)Geiping, Bauermeister, Dr{\"o}ge, and Moeller}]{geiping2020inverting}
\bibinfo{author}{J.~Geiping}, \bibinfo{author}{H.~Bauermeister}, \bibinfo{author}{H.~Dr{\"o}ge}, \bibinfo{author}{M.~Moeller},
\newblock \bibinfo{title}{Inverting gradients-how easy is it to break privacy in federated learning?},
\newblock \bibinfo{journal}{Advances in neural information processing systems} \bibinfo{volume}{33} (\bibinfo{year}{2020}) \bibinfo{pages}{16937--16947}.
%Type = Inproceedings
\bibitem[{Yin et~al.(2021)Yin, Mallya, Vahdat, Alvarez, Kautz, and Molchanov}]{yin2021see}
\bibinfo{author}{H.~Yin}, \bibinfo{author}{A.~Mallya}, \bibinfo{author}{A.~Vahdat}, \bibinfo{author}{J.~M. Alvarez}, \bibinfo{author}{J.~Kautz}, \bibinfo{author}{P.~Molchanov},
\newblock \bibinfo{title}{See through gradients: Image batch recovery via gradinversion},
\newblock in: \bibinfo{booktitle}{Proceedings of the IEEE/CVF Conference on Computer Vision and Pattern Recognition}, \bibinfo{year}{2021}, pp. \bibinfo{pages}{16337--16346}.
%Type = Inproceedings
\bibitem[{Dwork et~al.(2006)Dwork, McSherry, Nissim, and Smith}]{dwork2006calibrating}
\bibinfo{author}{C.~Dwork}, \bibinfo{author}{F.~McSherry}, \bibinfo{author}{K.~Nissim}, \bibinfo{author}{A.~Smith},
\newblock \bibinfo{title}{Calibrating noise to sensitivity in private data analysis},
\newblock in: \bibinfo{booktitle}{Theory of cryptography conference}, \bibinfo{organization}{Springer}, \bibinfo{year}{2006}, pp. \bibinfo{pages}{265--284}.
%Type = Inproceedings
\bibitem[{Abadi et~al.(2016)Abadi, Chu, Goodfellow, McMahan, Mironov, Talwar, and Zhang}]{abadi2016deep}
\bibinfo{author}{M.~Abadi}, \bibinfo{author}{A.~Chu}, \bibinfo{author}{I.~Goodfellow}, \bibinfo{author}{H.~B. McMahan}, \bibinfo{author}{I.~Mironov}, \bibinfo{author}{K.~Talwar}, \bibinfo{author}{L.~Zhang},
\newblock \bibinfo{title}{Deep learning with differential privacy},
\newblock in: \bibinfo{booktitle}{Proceedings of the 2016 ACM SIGSAC conference on computer and communications security}, \bibinfo{year}{2016}, pp. \bibinfo{pages}{308--318}.
%Type = Inproceedings
\bibitem[{Wang et~al.(2019)Wang, Song, Zhang, Song, Wang, and Qi}]{wang2019beyond}
\bibinfo{author}{Z.~Wang}, \bibinfo{author}{M.~Song}, \bibinfo{author}{Z.~Zhang}, \bibinfo{author}{Y.~Song}, \bibinfo{author}{Q.~Wang}, \bibinfo{author}{H.~Qi},
\newblock \bibinfo{title}{Beyond inferring class representatives: User-level privacy leakage from federated learning},
\newblock in: \bibinfo{booktitle}{IEEE INFOCOM 2019-IEEE conference on computer communications}, \bibinfo{organization}{IEEE}, \bibinfo{year}{2019}, pp. \bibinfo{pages}{2512--2520}.
%Type = Article
\bibitem[{Li et~al.(2020)Li, Sahu, Talwalkar, and Smith}]{li2020federated}
\bibinfo{author}{T.~Li}, \bibinfo{author}{A.~K. Sahu}, \bibinfo{author}{A.~Talwalkar}, \bibinfo{author}{V.~Smith},
\newblock \bibinfo{title}{Federated learning: Challenges, methods, and future directions},
\newblock \bibinfo{journal}{IEEE signal processing magazine} \bibinfo{volume}{37} (\bibinfo{year}{2020}) \bibinfo{pages}{50--60}.
%Type = Article
\bibitem[{Wang et~al.(2020)Wang, Liu, Liang, Joshi, and Poor}]{wang2020tackling}
\bibinfo{author}{J.~Wang}, \bibinfo{author}{Q.~Liu}, \bibinfo{author}{H.~Liang}, \bibinfo{author}{G.~Joshi}, \bibinfo{author}{H.~V. Poor},
\newblock \bibinfo{title}{Tackling the objective inconsistency problem in heterogeneous federated optimization},
\newblock \bibinfo{journal}{Advances in neural information processing systems} \bibinfo{volume}{33} (\bibinfo{year}{2020}) \bibinfo{pages}{7611--7623}.
%Type = Article
\bibitem[{Girgis et~al.(2021)Girgis, Data, Diggavi, Kairouz, and Suresh}]{girgis2021shuffled}
\bibinfo{author}{A.~M. Girgis}, \bibinfo{author}{D.~Data}, \bibinfo{author}{S.~Diggavi}, \bibinfo{author}{P.~Kairouz}, \bibinfo{author}{A.~T. Suresh},
\newblock \bibinfo{title}{Shuffled model of federated learning: Privacy, accuracy and communication trade-offs},
\newblock \bibinfo{journal}{IEEE journal on selected areas in information theory} \bibinfo{volume}{2} (\bibinfo{year}{2021}) \bibinfo{pages}{464--478}.
%Type = Article
\bibitem[{Zhang et~al.(2022)Zhang, Gu, Fan, Chen, and Yang}]{zhang2022no}
\bibinfo{author}{X.~Zhang}, \bibinfo{author}{H.~Gu}, \bibinfo{author}{L.~Fan}, \bibinfo{author}{K.~Chen}, \bibinfo{author}{Q.~Yang},
\newblock \bibinfo{title}{No free lunch theorem for security and utility in federated learning},
\newblock \bibinfo{journal}{arXiv preprint arXiv:2203.05816}  (\bibinfo{year}{2022}).
%Type = Article
\bibitem[{Mitchell et~al.(2022)Mitchell, Ball{\'e}, Charles, and Kone{\v{c}}n{\`y}}]{mitchell2022optimizing}
\bibinfo{author}{N.~Mitchell}, \bibinfo{author}{J.~Ball{\'e}}, \bibinfo{author}{Z.~Charles}, \bibinfo{author}{J.~Kone{\v{c}}n{\`y}},
\newblock \bibinfo{title}{Optimizing the communication-accuracy trade-off in federated learning with rate-distortion theory},
\newblock \bibinfo{journal}{arXiv preprint arXiv:2201.02664}  (\bibinfo{year}{2022}).
%Type = Article
\bibitem[{Zhang et~al.(2023{\natexlab{a}})Zhang, Chen, and Yang}]{zhang2023towards}
\bibinfo{author}{X.~Zhang}, \bibinfo{author}{K.~Chen}, \bibinfo{author}{Q.~Yang},
\newblock \bibinfo{title}{Towards achieving near-optimal utility for privacy-preserving federated learning via data generation and parameter distortion},
\newblock \bibinfo{journal}{arXiv preprint arXiv:2305.04288}  (\bibinfo{year}{2023}{\natexlab{a}}).
%Type = Article
\bibitem[{Zhang et~al.(2023{\natexlab{b}})Zhang, Kang, Chen, Fan, and Yang}]{zhang2023trading}
\bibinfo{author}{X.~Zhang}, \bibinfo{author}{Y.~Kang}, \bibinfo{author}{K.~Chen}, \bibinfo{author}{L.~Fan}, \bibinfo{author}{Q.~Yang},
\newblock \bibinfo{title}{Trading off privacy, utility, and efficiency in federated learning},
\newblock \bibinfo{journal}{ACM Transactions on Intelligent Systems and Technology} \bibinfo{volume}{14} (\bibinfo{year}{2023}{\natexlab{b}}) \bibinfo{pages}{1--32}.
%Type = Article
\bibitem[{Zhang et~al.(2024)Zhang, Kang, Fan, Chen, and Yang}]{zhang2024meta}
\bibinfo{author}{X.~Zhang}, \bibinfo{author}{Y.~Kang}, \bibinfo{author}{L.~Fan}, \bibinfo{author}{K.~Chen}, \bibinfo{author}{Q.~Yang},
\newblock \bibinfo{title}{A meta-learning framework for tuning parameters of protection mechanisms in trustworthy federated learning},
\newblock \bibinfo{journal}{ACM Transactions on Intelligent Systems and Technology} \bibinfo{volume}{15} (\bibinfo{year}{2024}) \bibinfo{pages}{1--36}.
%Type = Article
\bibitem[{Zhang et~al.(2023)Zhang, Fan, Wang, Li, Chen, and Yang}]{zhang2023game}
\bibinfo{author}{X.~Zhang}, \bibinfo{author}{L.~Fan}, \bibinfo{author}{S.~Wang}, \bibinfo{author}{W.~Li}, \bibinfo{author}{K.~Chen}, \bibinfo{author}{Q.~Yang},
\newblock \bibinfo{title}{A game-theoretic framework for federated learning},
\newblock \bibinfo{journal}{arXiv preprint arXiv:2304.05836}  (\bibinfo{year}{2023}).
%Type = Article
\bibitem[{He et~al.(2024)He, Chen, and Zhang}]{he2024reinforcement}
\bibinfo{author}{J.~He}, \bibinfo{author}{W.~Chen}, \bibinfo{author}{X.~Zhang},
\newblock \bibinfo{title}{Reinforcement learning as a catalyst for robust and fair federated learning: Deciphering the dynamics of client contributions},
\newblock \bibinfo{journal}{arXiv preprint arXiv} \bibinfo{volume}{2402} (\bibinfo{year}{2024}).
%Type = Article
\bibitem[{Zhang et~al.(2024)Zhang, Fei, Chen, and Jin}]{zhang2024deciphering}
\bibinfo{author}{X.~Zhang}, \bibinfo{author}{Y.~Fei}, \bibinfo{author}{W.~Chen}, \bibinfo{author}{H.~Jin},
\newblock \bibinfo{title}{Deciphering the interplay between local differential privacy, average bayesian privacy, and maximum bayesian privacy},
\newblock \bibinfo{journal}{arXiv preprint arXiv:2403.16591}  (\bibinfo{year}{2024}).
%Type = Inproceedings
\bibitem[{Nasr et~al.(2019)Nasr, Shokri, and Houmansadr}]{nasr2019comprehensive}
\bibinfo{author}{M.~Nasr}, \bibinfo{author}{R.~Shokri}, \bibinfo{author}{A.~Houmansadr},
\newblock \bibinfo{title}{Comprehensive privacy analysis of deep learning: Passive and active white-box inference attacks against centralized and federated learning},
\newblock in: \bibinfo{booktitle}{2019 IEEE symposium on security and privacy (SP)}, \bibinfo{organization}{IEEE}, \bibinfo{year}{2019}, pp. \bibinfo{pages}{739--753}.
%Type = Inproceedings
\bibitem[{Melis et~al.(2019)Melis, Song, De~Cristofaro, and Shmatikov}]{melis2019exploiting}
\bibinfo{author}{L.~Melis}, \bibinfo{author}{C.~Song}, \bibinfo{author}{E.~De~Cristofaro}, \bibinfo{author}{V.~Shmatikov},
\newblock \bibinfo{title}{Exploiting unintended feature leakage in collaborative learning},
\newblock in: \bibinfo{booktitle}{2019 IEEE symposium on security and privacy (SP)}, \bibinfo{organization}{IEEE}, \bibinfo{year}{2019}, pp. \bibinfo{pages}{691--706}.
%Type = Inproceedings
\bibitem[{Fredrikson et~al.(2015)Fredrikson, Jha, and Ristenpart}]{fredrikson2015model}
\bibinfo{author}{M.~Fredrikson}, \bibinfo{author}{S.~Jha}, \bibinfo{author}{T.~Ristenpart},
\newblock \bibinfo{title}{Model inversion attacks that exploit confidence information and basic countermeasures},
\newblock in: \bibinfo{booktitle}{Proceedings of the 22nd ACM SIGSAC conference on computer and communications security}, \bibinfo{year}{2015}, pp. \bibinfo{pages}{1322--1333}.
%Type = Article
\bibitem[{Bonawitz et~al.(2016)Bonawitz, Ivanov, Kreuter, Marcedone, McMahan, Patel, Ramage, Segal, and Seth}]{bonawitz2016practical}
\bibinfo{author}{K.~Bonawitz}, \bibinfo{author}{V.~Ivanov}, \bibinfo{author}{B.~Kreuter}, \bibinfo{author}{A.~Marcedone}, \bibinfo{author}{H.~B. McMahan}, \bibinfo{author}{S.~Patel}, \bibinfo{author}{D.~Ramage}, \bibinfo{author}{A.~Segal}, \bibinfo{author}{K.~Seth},
\newblock \bibinfo{title}{Practical secure aggregation for federated learning on user-held data},
\newblock \bibinfo{journal}{arXiv preprint arXiv:1611.04482}  (\bibinfo{year}{2016}).
%Type = Inproceedings
\bibitem[{Haddadpour et~al.(2021)Haddadpour, Kamani, Mokhtari, and Mahdavi}]{haddadpour2021federated}
\bibinfo{author}{F.~Haddadpour}, \bibinfo{author}{M.~M. Kamani}, \bibinfo{author}{A.~Mokhtari}, \bibinfo{author}{M.~Mahdavi},
\newblock \bibinfo{title}{Federated learning with compression: Unified analysis and sharp guarantees},
\newblock in: \bibinfo{booktitle}{International Conference on Artificial Intelligence and Statistics}, \bibinfo{organization}{PMLR}, \bibinfo{year}{2021}, pp. \bibinfo{pages}{2350--2358}.
%Type = Article
\bibitem[{Albasyoni et~al.(2020)Albasyoni, Safaryan, Condat, and Richt{\'a}rik}]{albasyoni2020optimal}
\bibinfo{author}{A.~Albasyoni}, \bibinfo{author}{M.~Safaryan}, \bibinfo{author}{L.~Condat}, \bibinfo{author}{P.~Richt{\'a}rik},
\newblock \bibinfo{title}{Optimal gradient compression for distributed and federated learning},
\newblock \bibinfo{journal}{arXiv preprint arXiv:2010.03246}  (\bibinfo{year}{2020}).
%Type = Article
\bibitem[{Geyer et~al.(2017)Geyer, Klein, and Nabi}]{geyer2017differentially}
\bibinfo{author}{R.~C. Geyer}, \bibinfo{author}{T.~Klein}, \bibinfo{author}{M.~Nabi},
\newblock \bibinfo{title}{Differentially private federated learning: A client level perspective},
\newblock \bibinfo{journal}{arXiv preprint arXiv:1712.07557}  (\bibinfo{year}{2017}).
%Type = Inproceedings
\bibitem[{Wu et~al.(2020)Wu, Farokhi, Smith, and Kaafar}]{wu2020value}
\bibinfo{author}{N.~Wu}, \bibinfo{author}{F.~Farokhi}, \bibinfo{author}{D.~Smith}, \bibinfo{author}{M.~A. Kaafar},
\newblock \bibinfo{title}{The value of collaboration in convex machine learning with differential privacy},
\newblock in: \bibinfo{booktitle}{2020 IEEE Symposium on Security and Privacy (SP)}, \bibinfo{organization}{IEEE}, \bibinfo{year}{2020}, pp. \bibinfo{pages}{304--317}.
%Type = Inproceedings
\bibitem[{Han et~al.(2020)Han, Wang, and Leung}]{han2020adaptive}
\bibinfo{author}{P.~Han}, \bibinfo{author}{S.~Wang}, \bibinfo{author}{K.~K. Leung},
\newblock \bibinfo{title}{Adaptive gradient sparsification for efficient federated learning: An online learning approach},
\newblock in: \bibinfo{booktitle}{2020 IEEE 40th international conference on distributed computing systems (ICDCS)}, \bibinfo{organization}{IEEE}, \bibinfo{year}{2020}, pp. \bibinfo{pages}{300--310}.
%Type = Article
\bibitem[{Zhang et~al.(2023{\natexlab{a}})Zhang, Li, Chen, Xia, and Yang}]{zhang2023theoretically}
\bibinfo{author}{X.~Zhang}, \bibinfo{author}{W.~Li}, \bibinfo{author}{K.~Chen}, \bibinfo{author}{S.~Xia}, \bibinfo{author}{Q.~Yang},
\newblock \bibinfo{title}{Theoretically principled federated learning for balancing privacy and utility},
\newblock \bibinfo{journal}{arXiv preprint arXiv:2305.15148}  (\bibinfo{year}{2023}{\natexlab{a}}).
%Type = Article
\bibitem[{Zhang et~al.(2023{\natexlab{b}})Zhang, Huang, Fan, Chen, and Yang}]{zhang2023probably}
\bibinfo{author}{X.~Zhang}, \bibinfo{author}{A.~Huang}, \bibinfo{author}{L.~Fan}, \bibinfo{author}{K.~Chen}, \bibinfo{author}{Q.~Yang},
\newblock \bibinfo{title}{Probably approximately correct federated learning},
\newblock \bibinfo{journal}{arXiv preprint arXiv:2304.04641}  (\bibinfo{year}{2023}{\natexlab{b}}).
%Type = Article
\bibitem[{Zhang et~al.(2024)Zhang, Xu, and Chen}]{zhang2024unified}
\bibinfo{author}{X.~Zhang}, \bibinfo{author}{M.~Xu}, \bibinfo{author}{W.~Chen},
\newblock \bibinfo{title}{A unified learn-to-distort-data framework for privacy-utility trade-off in trustworthy federated learning},
\newblock \bibinfo{journal}{arXiv preprint arXiv:2407.04751}  (\bibinfo{year}{2024}).
%Type = Book
\bibitem[{Royden and Fitzpatrick(1988)}]{royden1988real}
\bibinfo{author}{H.~L. Royden}, \bibinfo{author}{P.~Fitzpatrick}, \bibinfo{title}{Real analysis}, volume~\bibinfo{volume}{32}, \bibinfo{publisher}{Macmillan New York}, \bibinfo{year}{1988}.
%Type = Article
\bibitem[{Duchi et~al.(2011)Duchi, Hazan, and Singer}]{duchi2011adaptive}
\bibinfo{author}{J.~Duchi}, \bibinfo{author}{E.~Hazan}, \bibinfo{author}{Y.~Singer},
\newblock \bibinfo{title}{Adaptive subgradient methods for online learning and stochastic optimization.},
\newblock \bibinfo{journal}{Journal of machine learning research} \bibinfo{volume}{12} (\bibinfo{year}{2011}).
%Type = Article
\bibitem[{Kingma and Ba(2014)}]{kingma2014adam}
\bibinfo{author}{D.~P. Kingma}, \bibinfo{author}{J.~Ba},
\newblock \bibinfo{title}{Adam: A method for stochastic optimization},
\newblock \bibinfo{journal}{arXiv preprint arXiv:1412.6980}  (\bibinfo{year}{2014}).
%Type = Incollection
\bibitem[{McDiarmid(1989)}]{mcdiarmid1989}
\bibinfo{author}{C.~McDiarmid},
\newblock \bibinfo{title}{On the method of bounded differences},
\newblock in: \bibinfo{editor}{J.~Siemons} (Ed.), \bibinfo{booktitle}{Surveys in Combinatorics}, volume \bibinfo{volume}{141} of \textit{\bibinfo{series}{London Mathematical Society Lecture Note Series}}, \bibinfo{publisher}{Cambridge University Press}, \bibinfo{year}{1989}, pp. \bibinfo{pages}{148--188}.
%Type = Article
\bibitem[{LeCun et~al.(1998)LeCun, Bottou, Bengio, and Haffner}]{lecun1998gradient}
\bibinfo{author}{Y.~LeCun}, \bibinfo{author}{L.~Bottou}, \bibinfo{author}{Y.~Bengio}, \bibinfo{author}{P.~Haffner},
\newblock \bibinfo{title}{Gradient-based learning applied to document recognition},
\newblock \bibinfo{journal}{Proceedings of the IEEE} \bibinfo{volume}{86} (\bibinfo{year}{1998}) \bibinfo{pages}{2278--2324}.
%Type = Article
\bibitem[{Xiao et~al.(2017)Xiao, Rasul, and Vollgraf}]{xiao2017fashion}
\bibinfo{author}{H.~Xiao}, \bibinfo{author}{K.~Rasul}, \bibinfo{author}{R.~Vollgraf},
\newblock \bibinfo{title}{Fashion-mnist: a novel image dataset for benchmarking machine learning algorithms},
\newblock \bibinfo{journal}{arXiv preprint arXiv:1708.07747}  (\bibinfo{year}{2017}).
%Type = Techreport
\bibitem[{Krizhevsky(2009)}]{krizhevsky2009learning}
\bibinfo{author}{A.~Krizhevsky}, \bibinfo{title}{Learning Multiple Layers of Features from Tiny Images}, \bibinfo{type}{Technical Report}, University of Toronto, \bibinfo{year}{2009}. \bibinfo{note}{Technical Report}.
%Type = Article
\bibitem[{Wang et~al.(2004)Wang, Bovik, Sheikh, and Simoncelli}]{wang2004image}
\bibinfo{author}{Z.~Wang}, \bibinfo{author}{A.~C. Bovik}, \bibinfo{author}{H.~R. Sheikh}, \bibinfo{author}{E.~P. Simoncelli},
\newblock \bibinfo{title}{Image quality assessment: from error visibility to structural similarity},
\newblock \bibinfo{journal}{IEEE transactions on image processing} \bibinfo{volume}{13} (\bibinfo{year}{2004}) \bibinfo{pages}{600--612}.
%Type = Inproceedings
\bibitem[{Zhang et~al.(2018)Zhang, Isola, Efros, Shechtman, and Wang}]{zhang2018unreasonable}
\bibinfo{author}{R.~Zhang}, \bibinfo{author}{P.~Isola}, \bibinfo{author}{A.~A. Efros}, \bibinfo{author}{E.~Shechtman}, \bibinfo{author}{O.~Wang},
\newblock \bibinfo{title}{The unreasonable effectiveness of deep features as a perceptual metric},
\newblock in: \bibinfo{booktitle}{Proceedings of the IEEE conference on computer vision and pattern recognition}, \bibinfo{year}{2018}, pp. \bibinfo{pages}{586--595}.
%Type = Article
\bibitem[{Petrov et~al.(2024)Petrov, Dimitrov, Baader, M{\"u}ller, and Vechev}]{petrov2024dager}
\bibinfo{author}{I.~Petrov}, \bibinfo{author}{D.~I. Dimitrov}, \bibinfo{author}{M.~Baader}, \bibinfo{author}{M.~N. M{\"u}ller}, \bibinfo{author}{M.~Vechev},
\newblock \bibinfo{title}{Dager: Exact gradient inversion for large language models},
\newblock \bibinfo{journal}{Advances in neural information processing systems} \bibinfo{volume}{37} (\bibinfo{year}{2024}) \bibinfo{pages}{87801--87830}.
%Type = Inproceedings
\bibitem[{Sami et~al.(2025)Sami, Sen, Roy-Chowdhury, Krishnamurthy, and Guler}]{sami2025gradient}
\bibinfo{author}{H.~U. Sami}, \bibinfo{author}{S.~Sen}, \bibinfo{author}{A.~K. Roy-Chowdhury}, \bibinfo{author}{S.~V. Krishnamurthy}, \bibinfo{author}{B.~Guler},
\newblock \bibinfo{title}{Gradient inversion attacks on parameter-efficient fine-tuning},
\newblock in: \bibinfo{booktitle}{Proceedings of the IEEE/CVF Conference on Computer Vision and Pattern Recognition}, \bibinfo{year}{2025}, pp. \bibinfo{pages}{10224--10234}.
%Type = Article
\bibitem[{Xie et~al.(2025)Xie, He, Li, Jia, and Ji}]{xie2025recit}
\bibinfo{author}{J.~Xie}, \bibinfo{author}{R.~He}, \bibinfo{author}{S.~Li}, \bibinfo{author}{X.~Jia}, \bibinfo{author}{S.~Ji},
\newblock \bibinfo{title}{Recit: Reconstructing full private data from gradient in parameter-efficient fine-tuning of large language models},
\newblock \bibinfo{journal}{arXiv preprint arXiv:2504.20570}  (\bibinfo{year}{2025}).
%Type = Inproceedings
\bibitem[{Li et~al.(2024)Li, Xu, and Dras}]{li2024seeing}
\bibinfo{author}{W.~Li}, \bibinfo{author}{Q.~Xu}, \bibinfo{author}{M.~Dras},
\newblock \bibinfo{title}{Seeing the forest through the trees: Data leakage from partial transformer gradients},
\newblock in: \bibinfo{booktitle}{Proceedings of the 2024 Conference on Empirical Methods in Natural Language Processing}, \bibinfo{year}{2024}, pp. \bibinfo{pages}{4786--4798}.

\end{thebibliography}

\newpage
\onecolumn
\appendix
\section{Analysis for \pref{thm: not_uniquely_determined}}
\label{sec: thm: not_uniquely_determined}

\begin{thm}\label{thm: not_uniquely_determined_app}
Let $d$ be the model parameter dimension and $p$ the dimension
of a single input sample. Fix the labels
$\{y_b^k\}_{b=1}^{B}$, the initial model parameters, the local
training configuration, and any auxiliary random seeds.
Write
\[
\mathbf{x}^k
=
\bigl((x_1^k)^\top,\ldots,(x_B^k)^\top\bigr)^\top
\in \mathbb{R}^{Bp}.
\]
Consider the local-update mapping
\[
F:U\longrightarrow\mathbb{R}^{d},
\qquad
F(\mathbf{z})
=
\Delta\theta_t^k
\bigl(\{z_b,y_b^k\}_{b=1}^{B}\bigr),
\]
where $U\subset\mathbb{R}^{Bp}$ is an open neighborhood of
$\mathbf{x}^k$ contained in the admissible input domain.
Assume that $F$ is continuously differentiable on $U$ and that
its Jacobian has constant rank on $U$:
\[
\operatorname{rank} DF(\mathbf{z})
=
r<Bp,
\qquad
\mathbf{z}\in U.
\]

Then, on a sufficiently small neighborhood
$V\subset U$ of $\mathbf{x}^k$, the level set
\[
\mathcal{M}
=
\left\{
\mathbf{z}\in V:
F(\mathbf{z})=F(\mathbf{x}^k)
\right\}
\]
is a $C^1$ submanifold of dimension $Bp-r>0$.
Consequently, there exist arbitrarily small nonzero
perturbations
\[
\delta\mathbf{x}^k
=
\bigl((\delta x_1^k)^\top,\ldots,
      (\delta x_B^k)^\top\bigr)^\top
\]
such that $\mathbf{x}^k+\delta\mathbf{x}^k\in U$ and
\[
\begin{aligned}
&\Delta\theta_t^k
\bigl(\{x_b^k+\delta x_b^k,y_b^k\}_{b=1}^{B}\bigr) \\
&\qquad =
\Delta\theta_t^k
\bigl(\{x_b^k,y_b^k\}_{b=1}^{B}\bigr).
\end{aligned}
\]
Thus, the private input batch is not locally identifiable
from this model update, even when its labels are fixed.
\end{thm}

\begin{proof}
Let $n=Bp$ and let $F$ be the local-update mapping in the
theorem. By the rank--nullity theorem,
\[
\dim\ker DF(\mathbf{x}^k)
=
n-r>0.
\]
To obtain exact equality of model updates, we use the
constant rank theorem rather than a first-order
approximation.
Because $F$ is $C^1$ and has constant rank $r$ near
$\mathbf{x}^k$, there exist neighborhoods $V$ of
$\mathbf{x}^k$ and $W$ of $F(\mathbf{x}^k)$, together with
$C^1$ coordinate diffeomorphisms $\phi$ and $\psi$, such that
\[
\phi(\mathbf{x}^k)=0,
\qquad
\psi(F(\mathbf{x}^k))=0,
\]
and
\[
(\psi\circ F\circ\phi^{-1})(u_1,\ldots,u_n)
=
(u_1,\ldots,u_r,0,\ldots,0)
\]
for all $(u_1,\ldots,u_n)$ sufficiently close to the origin.

Since $r<n$, let $e_{r+1}$ be the $(r+1)$-st standard
basis vector of $\mathbb{R}^{n}$ and define
\[
\gamma(s)=\phi^{-1}(s e_{r+1})
\]
for sufficiently small $s$. The first $r$ coordinates
of $s e_{r+1}$ are zero, and hence
\[
\psi(F(\gamma(s)))
=
0
=
\psi(F(\mathbf{x}^k)).
\]
The injectivity of $\psi$ implies
\[
F(\gamma(s))=F(\mathbf{x}^k).
\]
For every sufficiently small $s\neq0$, the injectivity
of $\phi$ gives $\gamma(s)\neq\mathbf{x}^k$. Therefore,
\[
\begin{aligned}
\delta\mathbf{x}^k(s)
&=\gamma(s)-\mathbf{x}^k\neq0, \\
F(\mathbf{x}^k+\delta\mathbf{x}^k(s))
&=F(\mathbf{x}^k).
\end{aligned}
\]
Moreover, $\delta\mathbf{x}^k(s)\to0$ as $s\to0$.
Thus, there are arbitrarily close but distinct admissible
batches producing exactly the same update. Equivalently,
the local level set through $\mathbf{x}^k$ has dimension
$Bp-r>0$, which proves the claimed local non-identifiability.
\end{proof}
\section{Analysis for \pref{thm: quantitative_dbp}}
\label{sec: thm: quantitative_dbp}

\begin{thm}\label{thm: quantitative_dbp_app}
Consider horizontal federated learning with model parameter
dimension $d$ and input dimension $p$. Fix the labels,
initial model parameters, local training configuration,
and auxiliary random seeds.
Assume that the local-update mapping is continuously
differentiable on an open neighborhood of the true input
batch contained in the admissible input domain, and that
its Jacobian has constant rank on this neighborhood.

If
\[
d<Bp,
\]
then the true input batch is not locally identifiable from
the observed model update. In particular, arbitrarily close
but distinct admissible batches with the same labels
produce exactly the same update.
Thus,
\[
B\geq \left\lfloor\frac{d}{p}\right\rfloor+1
\]
is a sufficient dimension condition for this
non-identifiability conclusion under the stated regularity
assumptions.
\end{thm}

\begin{proof}
Let $J$ be the Jacobian of the full local-update mapping.
Since $J\in\mathbb{R}^{d\times Bp}$ and $d<Bp$,
\[
\operatorname{rank}(J)\leq d<Bp.
\]
The continuous-differentiability and local constant-rank
conditions are assumed in the theorem. Therefore,
Theorem~\ref{thm: not_uniquely_determined} applies and
establishes local non-identifiability of the true input
batch from the observed update.
\end{proof}
\section{Analysis for \pref{thm:privacy_leakage_upper_bound}}
\label{sec: thm:privacy_leakage_upper_bound}

\begin{thm}[Upper Bound for Privacy Leakage]
\label{thm:privacy_leakage_upper_bound_app}
Let the setting and assumptions of
\pref{thm:privacy_leakage_upper_bound} hold. Assume that
\begin{equation}
\overline{\Delta}^k
\ge
\frac{2c_2c_bE}{c_a\sqrt{T}},
\label{eq:privacy_distortion_condition}
\end{equation}
where $c_a$, $c_b$, and $c_2$ are the constants introduced in
\pref{assump: two-sided Lipschitz} and
\pref{assump: bounds_for_optimization_alg}.
Let $B$ denote the local batch size, $E$ the number of local
epochs, and $T$ the number of optimization rounds.

Then, with probability at least $1-\frac{1}{B}$, the privacy
leakage of client $k$ satisfies
\begin{equation}
\widehat{\epsilon}_{p,B}^k(X)
\le
1
+
\sqrt{\frac{c'^{\,2}(\ln 2+\ln B)}{2B}}
-
\frac{c_a}{2D}\overline{\Delta}^k,
\end{equation}
where $D$ is the upper bound on the distance between two
data samples, and $c'$ is the bounded-difference constant
introduced in Lemma~\ref{lem:mcdiarmid_concentration}.
\end{thm}

\begin{proof}
Let $\mathcal{D}_{\mathrm{local}}^{k}$ denote the local
dataset of client $k$, and let
$\mathcal{D}_{\mathrm{batch}}^{k}
=\{(\mathbf{x}_{i}^{k},y_i^k)\}_{i=1}^{B}$
be an i.i.d. batch under the fixed-state sampling model
specified in Section~\ref{sec:average_case_optimization},
denoted by $X$. Let $\mathbf{x}_{t,i}^{k}$ denote the
reconstruction of $\mathbf{x}_{i}^{k}$ produced by the
attacker at optimization round $t$, and let
$\tilde{\mathbf{x}}_{i}^{k}$ denote the surrogate sample
associated with the distorted gradient. Throughout this proof, we write $\Delta^k=\Delta^k(X)$ for the
realized distortion of the sampled batch $X$, and
$\overline{\Delta}^k=\mathbb{E}_X[\Delta^k(X)]$ for its expectation
under the fixed-state sampling model.

We first consider the batch-average reconstruction error
\begin{equation}
\frac{1}{B}\sum_{i=1}^{B}
\frac{1}{T}\sum_{t=1}^{T}
\left\|
\mathbf{x}_{t,i}^{k}-\mathbf{x}_{i}^{k}
\right\|.
\end{equation}
By the triangle inequality,
\begin{equation}
\begin{aligned}
&\frac{1}{B}\sum_{i=1}^{B}
\frac{1}{T}\sum_{t=1}^{T}
\left\|
\mathbf{x}_{t,i}^{k}-\mathbf{x}_{i}^{k}
\right\| \\
&\ge
\frac{1}{B}\sum_{i=1}^{B}
\left\|
\tilde{\mathbf{x}}_{i}^{k}-\mathbf{x}_{i}^{k}
\right\| \\
&\quad -
\frac{1}{T}\sum_{t=1}^{T}
\frac{1}{B}\sum_{i=1}^{B}
\left\|
\mathbf{x}_{t,i}^{k}-\tilde{\mathbf{x}}_{i}^{k}
\right\|.
\end{aligned}
\label{eq:triangle_revised}
\end{equation}
According to the bi-Lipschitz condition in
\pref{assump: two-sided Lipschitz}, we have
\begin{align}
\left\|
\tilde{\mathbf{x}}_{i}^{k}-\mathbf{x}_{i}^{k}
\right\|
&\ge
c_a
\left\|
\nabla\ell(\theta,\tilde{\mathbf{x}}_{i}^{k},y_i^k)
-
\nabla\ell(\theta,\mathbf{x}_{i}^{k},y_i^k)
\right\|,
\\
\left\|
\mathbf{x}_{t,i}^{k}-\tilde{\mathbf{x}}_{i}^{k}
\right\|
&\le
c_b
\left\|
\nabla\ell(\theta,\mathbf{x}_{t,i}^{k},y_i^k)
-
\nabla\ell(\theta,\tilde{\mathbf{x}}_{i}^{k},y_i^k)
\right\|.
\end{align}
Substituting these inequalities into
\eqref{eq:triangle_revised} gives
\begin{equation}
\begin{aligned}
&\frac{1}{B}\sum_{i=1}^{B}
\frac{1}{T}\sum_{t=1}^{T}
\left\|
\mathbf{x}_{t,i}^{k}-\mathbf{x}_{i}^{k}
\right\| \\
&\ge
\frac{c_a}{B}\sum_{i=1}^{B}
\left\|
\nabla\ell(\theta,\tilde{\mathbf{x}}_{i}^{k},y_i^k)
-
\nabla\ell(\theta,\mathbf{x}_{i}^{k},y_i^k)
\right\| \\
&\quad -
\frac{c_b}{T}\sum_{t=1}^{T}
\frac{1}{B}\sum_{i=1}^{B}
\left\|
\nabla\ell(\theta,\mathbf{x}_{t,i}^{k},y_i^k)
-
\nabla\ell(\theta,\tilde{\mathbf{x}}_{i}^{k},y_i^k)
\right\|.
\end{aligned}
\label{eq:gradient_distance_revised}
\end{equation}

Using
\begin{equation}
\frac{1}{B}\sum_{i=1}^{B}\|\mathbf{a}_i\|
\ge
\left\|
\frac{1}{B}\sum_{i=1}^{B}\mathbf{a}_i
\right\|
\end{equation}
and the definition of the distortion extent $\Delta^k$,
the first term on the right-hand side satisfies
\begin{equation}
\begin{aligned}
&\frac{c_a}{B}\sum_{i=1}^{B}
\left\|
\nabla\ell(\theta,\tilde{\mathbf{x}}_{i}^{k},y_i^k)
-
\nabla\ell(\theta,\mathbf{x}_{i}^{k},y_i^k)
\right\| \\
&\ge
c_a
\left\|
\frac{1}{B}\sum_{i=1}^{B}
\left(
\nabla\ell(\theta,\tilde{\mathbf{x}}_{i}^{k},y_i^k)
-
\nabla\ell(\theta,\mathbf{x}_{i}^{k},y_i^k)
\right)
\right\| \\
&=
c_a\Delta^k.
\end{aligned}
\end{equation}
For the second term,
Assumption~\ref{assump: bounds_for_optimization_alg}
directly bounds the cumulative average of the per-sample
gradient-residual norms:
\begin{equation}
\begin{aligned}
&\sum_{t=1}^{T}
\frac{1}{B}\sum_{i=1}^{B}
\left\|
\nabla\ell(\theta,\mathbf{x}_{t,i}^{k},y_i^k)
-
\nabla\ell(\theta,\tilde{\mathbf{x}}_{i}^{k},y_i^k)
\right\| \\
&\le Ec_2\sqrt{T}.
\end{aligned}
\label{eq:regret_batch_bound}
\end{equation}
Consequently,
\begin{equation}
\begin{aligned}
&\frac{1}{T}\sum_{t=1}^{T}
\frac{1}{B}\sum_{i=1}^{B}
\left\|
\nabla\ell(\theta,\mathbf{x}_{t,i}^{k},y_i^k)
-
\nabla\ell(\theta,\tilde{\mathbf{x}}_{i}^{k},y_i^k)
\right\| \\
&\le \frac{Ec_2}{\sqrt{T}}.
\end{aligned}
\end{equation}
Combining the above inequalities yields
\begin{equation}
\frac{1}{B}\sum_{i=1}^{B}
\frac{1}{T}\sum_{t=1}^{T}
\left\|
\mathbf{x}_{t,i}^{k}-\mathbf{x}_{i}^{k}
\right\|
\ge
c_a\Delta^k
-
\frac{c_bEc_2}{\sqrt{T}}.
\label{eq:reconstruction_lower_bound}
\end{equation}

Define the normalized batch-average reconstruction error as
\begin{equation}
f(\mathbf{x}_1^k,\ldots,\mathbf{x}_B^k)
\triangleq
\frac{1}{B}\sum_{i=1}^{B}
\frac{1}{T}\sum_{t=1}^{T}
\frac{
\|\mathbf{x}_{t,i}^{k}-\mathbf{x}_{i}^{k}\|
}{D}.
\end{equation}
Taking expectation over the random sampling of the local
batch and using \eqref{eq:reconstruction_lower_bound},
we obtain
\begin{equation}
\mathbb{E}[f]
\ge
\frac{c_a}{D}\overline{\Delta}^k
-
\frac{c_bEc_2}{D\sqrt{T}}.
\label{eq:expected_reconstruction_lower_bound}
\end{equation}
For the current batch, define the realized leakage by
\begin{equation}
\widehat{\epsilon}_{p,B}^k(X)=1-f(X).
\end{equation}
Separately,
\eqref{eq:expected_reconstruction_lower_bound} and
\eqref{eq:privacy_distortion_condition} give
\[
\mathbb{E}_X[\widehat{\epsilon}_{p,B}^k(X)]
\le
1-\frac{c_a}{2D}\overline{\Delta}^k.
\]
This batch expectation is not identified with the
dataset-level metric in
Definition~\ref{defi:Privacy_Leakage}.

We next derive a high-probability bound around
$\mathbb{E}[f]$. According to
Lemma~\ref{lem:mcdiarmid_concentration}, changing one sample
in the batch changes $f$ by at most $\frac{c'}{B}$.
Hence, McDiarmid's inequality gives
\begin{equation}
\Pr\left(
\left|f-\mathbb{E}[f]\right|\ge\varepsilon
\right)
\le
2\exp\left(
-\frac{2B\varepsilon^2}{c'^{\,2}}
\right).
\end{equation}
Let
\begin{equation}
\varepsilon
=
\sqrt{\frac{c'^{\,2}(\ln 2+\ln B)}{2B}}.
\end{equation}
Then
\begin{equation}
2\exp\left(
-\frac{2B\varepsilon^2}{c'^{\,2}}
\right)
=
\frac{1}{B}.
\end{equation}
Thus, with probability at least $1-\frac{1}{B}$,
\begin{equation}
f
\ge
\mathbb{E}[f]
-
\sqrt{\frac{c'^{\,2}(\ln 2+\ln B)}{2B}}.
\end{equation}

Combining this concentration bound with
\eqref{eq:expected_reconstruction_lower_bound}, we have,
with probability at least $1-\frac{1}{B}$,
\begin{equation}
f
\ge
\frac{c_a}{D}\overline{\Delta}^k
-
\frac{c_bEc_2}{D\sqrt{T}}
-
\sqrt{\frac{c'^{\,2}(\ln 2+\ln B)}{2B}}.
\end{equation}
Equivalently,
\begin{equation}
\begin{aligned}
\widehat{\epsilon}_{p,B}^k(X)
&\le
1
-
\frac{c_a}{D}\overline{\Delta}^k
+
\frac{c_bEc_2}{D\sqrt{T}} \\
&\quad +
\sqrt{\frac{c'^{\,2}(\ln 2+\ln B)}{2B}}.
\end{aligned}
\label{eq:privacy_bound_before_simplification}
\end{equation}
Finally, by the condition
\begin{equation}
\overline{\Delta}^k
\ge
\frac{2c_2c_bE}{c_a\sqrt{T}},
\end{equation}
we have
\begin{equation}
\frac{c_bEc_2}{D\sqrt{T}}
\le
\frac{c_a}{2D}\overline{\Delta}^k.
\end{equation}
Substituting this inequality into
\eqref{eq:privacy_bound_before_simplification} gives
\begin{equation}
\boxed{
\widehat{\epsilon}_{p,B}^k(X)
\le
1
+
\sqrt{\frac{c'^{\,2}(\ln 2+\ln B)}{2B}}
-
\frac{c_a}{2D}\overline{\Delta}^k
}
\end{equation}
with probability at least $1-\frac{1}{B}$.
This completes the proof.
\end{proof}

\end{document}